\documentclass[12pt,a4paper]{article}
\usepackage{amsfonts}
\usepackage{amsmath}
\usepackage{amsthm}
\usepackage{amssymb}
\usepackage{mathrsfs}
\usepackage[numbers]{natbib}
\usepackage[fit]{truncate}
\usepackage{fullpage}
\usepackage{hyperref}
\usepackage{tikz}
\usepackage{microtype}
\usepackage{mathtools}
\usepackage{enumitem}
\mathtoolsset{showonlyrefs}

\usepackage[linesnumbered,tworuled,vlined]{algorithm2e}
\usepackage{algorithmic}

\usepackage{titlesec}
\titlelabel{\thetitle.\;} 

\usepackage{bbm}
\newcommand{\indic}[1]{{\mathbbm{1}}_{#1}}

\newcommand{\ifff}[0]{iff}
\newcommand{\ie}[0]{i.e.}
\newcommand{\as}[0]{a.s.}
\newcommand{\eg}[0]{e.g.}
\newcommand{\vanish}[1]{}

\newcommand{\tit}[1]{\textit{#1}}
\newcommand{\tbf}[1]{\textbf{#1}}

\newcommand{\defformat}[1]{\textnormal{({#1})}}
\newcommand{\APP}[0]{\defformat{APP}}
\newcommand{\NFLVR}[0]{\defformat{NFLVR}}
\newcommand{\NUPBR}[0]{\defformat{NUPBR}}
\newcommand{\NA}[0]{\defformat{NA}}
\newcommand{\ELMM}[0]{\defformat{ELMM}}
\newcommand{\MP}[0]{\epsilon_{\textnormal{MP}}}

\newcommand{\wt}[0]{\widetilde}
\newcommand{\mc}[0]{\mathcal}
\newcommand{\msf}[0]{\mathsf}

\newcommand{\ra}[0]{ \rightarrow }
\newcommand{\lra}[0]{ \longrightarrow }

\newcommand{\Lloc}[0]{ L^{\infty}_{\loc}((0,\infty)) }

\DeclareMathOperator{\loc}{loc}

\DeclareMathOperator{\Opt}{Opt}
\DeclareMathOperator{\Esp}{E}
\DeclareMathOperator{\Prob}{P}

\DeclareMathOperator{\Qrob}{Q}

\DeclareMathOperator{\IR}{\mathbb{R}}
\DeclareMathOperator{\IN}{\mathbb{N}}

\DeclareMathOperator{\bF}{\mathcal{F}}

\DeclareMathOperator{\bE}{\mathcal{E}}

\DeclareMathOperator{\dom}{dom}

\DeclareMathOperator{\Var}{Var}

\DeclareMathOperator{\Lop}{L}

\DeclareMathOperator*{\sgn}{sgn}

\DeclareMathOperator*{\MC}{MC}

\newcommand{\rd}{\mathrm{d}}
\newcommand{\vd}{\,\mathrm{d}}

\newcommand{\process}[1]{(#1)_{t\ge 0}}

\newcommand{\loct}[3]{L^{#2}_{#3}(#1)}

\newcommand{\xnorm}[2]{ \norm{#1}_{#2} }

\newcommand{\eqnumB}{\refstepcounter{equation}\textup{\tagform@{\theequation}}}

\newcommand{\braces}[1]{ ( #1 ) } 
\newcommand{\bigbraces}[1]{ \big( #1 \big) } 
\newcommand{\Bigbraces}[1]{ \Big( #1 \Big) } 
\newcommand{\biggbraces}[1]{ \bigg( #1 \bigg) } 

\newcommand{\sqbraces}[1]{ [#1 ] } 
\newcommand{\bigsqbraces}[1]{ \big[ #1 \big] }

\newcommand{\bigcubraces}[1]{ \big\{ #1 \big\}}

\newcommand{\norm}[1]{ \| #1  \|}

\newcommand{\qv}[1]{ \langle #1 \rangle }

\theoremstyle{plain}
\newtheorem{theorem}{Theorem}[section]
\newtheorem{corollary}[theorem]{Corollary}
\newtheorem{lemma}[theorem]{Lemma}

\newtheorem{proposition}[theorem]{Proposition}

\newtheorem*{lemmaOhne}{Lemma}

\theoremstyle{definition}
\newtheorem{definition}[theorem]{Definition}
\newtheorem{example}[theorem]{Example}

\newtheorem{remark}[theorem]{Remark}

\newtheorem{SA}[theorem]{Standing Assumption}

\usepackage{authblk}

\newcommand{\email}[1]{\href{mailto:{#1}}{{#1}}}

\usepackage{authblk}

\title{Pricing and hedging for a sticky diffusion}

\author[1]{Alexis Anagnostakis\thanks{\email{alexis.anagnostakis@univ-grenoble-alpes.fr}}}
\affil[1]{Université Grenoble-Alpes, CNRS, LJK, F-38000, Grenoble, France}

\date{}

\begin{document}
	
	\maketitle
	
	\vspace*{-1cm}
	
	\begin{abstract}
We introduce a financial market model featuring a risky asset whose price follows a sticky geometric Brownian motion and a riskless asset that grows with a constant interest rate $r\in \mathbb R $. 
We prove that this model satisfies No Arbitrage (NA) and No Free Lunch with Vanishing Risk (NFLVR) only when $r=0 $. 
Under this condition, we derive the corresponding arbitrage-free pricing equation, assess replicability and representation of the replication strategy. 
We then show that all locally bounded replicable payoffs for the standard Black--Scholes model are also replicable for the sticky model. 
Last, we evaluate via numerical experiments the impact of hedging in discrete time and of misrepresenting price stickiness.
	\end{abstract}

	\renewcommand{\thefootnote}{}

	\footnotetext{ 
		\textit{Date:} \date{\today} 
	}
	
	\renewcommand{\thefootnote}{\arabic{footnote}}


	\section{Introduction} 
	\label{sec_intro}
	
	The seminal work of Black and Scholes \cite{BlaSch73} stands as the founding result in the theory of derivatives valuation and forms the basis of most pricing models in mathematical finance. 
	The model assumes that the price of the risky asset is a geometric Brownian motion  and that the interest rate is constant.
	Under these conditions, it proposes a riskless replication (or hedging) strategy of the payoff of a European contingent claim through a self-financing portfolio composed of the underlying risky asset and the non--risky asset. 
	In accordance with the arbitrage--free pricing principle, the claim's price is equal to the cost of the replication strategy. 
	This allows effective market risk exposure management and fair pricing of  financial products.
	
	Most continuous models in mathematical finance extend the Black-Scholes framework to capture additional features of stock and derivatives price dynamics. The assumptions in the Black-Scholes model, such as constant volatility, drift, and interest rate, are unrealistic when applied to both historical and implied price behavior. Several models have been developed to address these limitations. Local volatility models~\cite{Dupire1997,dupire1994pricing} and stochastic volatility models~\cite{hagan2002managing,Heston1993,hull1987pricing,scott1987option,stein1991stock} account for the stochastic nature of volatility and provide a better fit for the  implied volatility structure. Models incorporating jumps (\eg{}~\cite{bates1996jumps}) and stochastic interest rates (\eg{}~\cite{bakshi1997empirical}) offer further refinements to market dynamics. In addition, skew-threshold dynamics and reflecting-threshold dynamics have been explored (see~\cite{corns2007skew,decamps2004applications,Rossello2012,buckner2024arbitrage} and references therein). Rough volatility models have been considered to better fit statistical observations~\cite{bayer2016pricing,gatheral2018volatility}. 
	Models addressing liquidity and transaction costs have been proposed~\cite{almgren2001optimal,leland1985option}, lifting additional assumptions of the Black-Scholes model.
	
	An effect that is observed on price series is the presence of thresholds where the dynamic seems to linger or even spend a positive amount of time there. 
	This is observed in the case of corporate takeovers, where the takeover price is considered to be fair by the market.
	This is mentioned in the Introduction of~\cite{Bass} as motivation for his work on the description of stochastic differential equations with a sticky point.
	An example of this occurrence is given in~\cite{criens2022separating} on
	Hansen Technologies Limited between June 7, 2021 and September 6, 2021, where a takeover offer was made and then canceled.
	Softer psychological barriers are also observed in other markets, especially in commodities (see for example gold and silver). 
	
	In this paper, we consider a variant of the Black-Scholes model where the price-dynamic $S$ of the risky asset is a sticky geometric Brownian motion.  
	The process $S$ is a diffusion that behaves like a geometric Brownian motion away for certain price threshold $\zeta>0 $, where it spends a positive amount time upon contact.
	We refer to this model as the \tit{sticky Black-Scholes model}. 
	We establish that there is \tit{No Arbitrage} \NA{} and \tit{No Free Lunch with Vanishing Risk} \NFLVR{} in the model if and only if the interest rate $r$ is zero.
	In this case, we derive the arbitrage--free pricing equation, assess the existence of a replication strategy and provide an explicit representation of it in terms of the solution of the pricing problem. 
	We then prove additional properties on the prices on claims and replicability.
	In particular, that the prices of claims with convex payoffs are decreasing and the class of locally bounded replicable payoffs is increasing with respect to the stickiness parameter.
	This yields that, for this model, one can consider all locally bounded payoff functions that can be replicated in the standard Black-Scholes model. 
	Last, we conduct numerical evaluations of discrete-time hedging and analyze the hedging error incurred when ignoring or misrepresenting price stickiness (model mismatch) as realized volatility or local volatility in a smooth model.

	The proof of sufficient and necessary condition for \NFLVR{} is done by constructing an arbitrage strategy in the case $r\not = 0 $ and by finding an \ELMM{} in the case $r=0 $. 
	The construction of the arbitrage strategy is based on the proof of \cite[Theorem~12.3.5]{Delbaen2006} and relies on the presence of a local time term in the discounted price process that cannot be offset by the
	 quadratic variation of the martingale part. 
	This technique has been previously used to demonstrate the existence of arbitrage in~\cite{buckner2024arbitrage}, for a diffusion with a reflecting price threshold, and in~\cite{Rossello2012}, for a diffusion with skew returns.
	We note that, unlike these models, in a sticky price model, the local time term appears only for the discounted price process
	when $r\not = 0 $. 
	
	In our analysis, we show that the sticky Black--Scholes model is an example of a model where there is no uniqueness of \tit{Equivalent Local Martingale Measure} \ELMM{} for the discounted price process and no uniqueness of payoff replication strategy.
	These peculiar effects arise from the ability to trade in the infinitesimal range of  $\indic{S_t = \zeta} \vd t$, the times when the price process is located at the sticky threshold. Due to the stickiness, this set has positive Lebesgue measure. The arbitrage strategies in the case $r \neq 0$ are necessarily partially or fully supported on the set $\indic{S_t = \zeta} \vd t$. In the case $r = 0$, the \ELMM{} and the payoff replication strategy are equivalence classes of processes that coincide over the range of $\indic{S_t \not= \zeta} \vd t$.

	First discovered by Feller~\cite{Fel52}, sticky diffusions are Markov processes with continuous sample paths that spend  positive amount of time at some threshold(s).
	They are characterized by the presence of atom(s) in their speed measures~\cite{Bass,BorSal} and can be expressed as time-delays of non-sticky diffusions (see~\cite{Ami,EngPes,Sal2017}). 
	The time-delay is proportional to the local time of the diffusion at the sticky threshold.  
	This leads in some cases to path-space representations similar to classical SDEs for absolutely continuous diffusions~\cite{EngPes,Sal2017} 
	(see also~\cite[Section~4]{Anagnostakis2022}).
	
	In finance, sticky diffusions have been studied to model other phenomena like   near zero short rate dynamics~\cite{Nie2,Zhang2022} and winning streak dynamics~\cite{Feng2020}.
	For other applications, see~\cite{BouRabee2020,CalFar,DavTru,Stell1991,ZhuJoh}.

	The paper is organized as follows. 
	In Section~\ref{sec_model}, we define the sticky Black-Scholes model. 
	In Section~\ref{sec_arbitrage}, we show that the sticky Black-Scholes model satisfies  \NFLVR{} \ifff{} $r=0$.
	Under this condition, we identify an equivalent local martingale measure for the price process. 
	In Section~\ref{sec_pricing} we first derive the arbitrage--free pricing equation
	for the case $r=0 $ and prove the existence of the associated payoff replication strategy.
	Section~\ref{sec_monotonicity} explores monotonicity properties in both prices and the classes of replicable payoffs.
	In Section~\ref{sec_numerical}, we numerically evaluate properties of discrete-time hedging and assess the impact of model mismatch.
	Finally, Section~\ref{sec_conclusion} provides concluding remarks.

	\paragraph{Acknowledgments.}
	 The author expresses sincere gratitude to David Criens for his valuable feedback and suggestions on the initial preprint. Special thanks go to Sara Mazzonetto and Denis Villemonais for their insightful discussions and helpful references during the early stages of this work. The author also appreciates the constructive comments from the Editor-in-Chief and the two anonymous referees.

	\section{The model}
	\label{sec_model}

	The standard Black-Scholes model makes the following assumptions:
	\begin{enumerate}
		[label= (\roman*), font=, labelsep=.5em]
		\itemsep-0.3em 
		\item infinite liquidity, divisibility of assets, and borrowing/lending capacity,
		\item frictionless market (no transaction costs, no borrowing or lending rate),
		\item possibility of continuous-time trading,
		\item constant interest rate $r$,
		\item \label{item_BD_geometric} geometric Brownian motion price dynamics.
	\end{enumerate}
	
	In the sticky Black-Scholes model, assumption~\ref{item_BD_geometric} is modified such that the price process of the risky asset $S$ follows a sticky geometric Brownian motion. 
	In particular, the prices of risky and non-risky assets $(S,S^{0})$ solve the following system: 
	\begin{equation}
		\label{eq_def_BS_diffusion}
		\begin{dcases}
			\vd S_t=   \indic{S_t \ne \zeta} S_t \left(\mu \vd t + \sigma \vd B_t\right), \\
			\indic{S_t = \zeta} \vd t = \rho \vd \loct{S}{\zeta}{t}, \\
			\vd S^{0}_t= r S^{0}_t \vd t,
		\end{dcases}
	\end{equation}
	where $\mu, r \in \mathbb{R}$, $\sigma, \zeta > 0$, $\rho > 0 $, where $B$ is a standard Brownian motion, and where $(\loct{S}{y}{t};\,t\ge 0,y>0) $ is the local time field of $S$.

	\begin{figure}[h!]
		\begin{center}
			\includegraphics[alt={Sample path of sticky and non-sticky geometric Brownian motions},width = 0.7\textwidth]{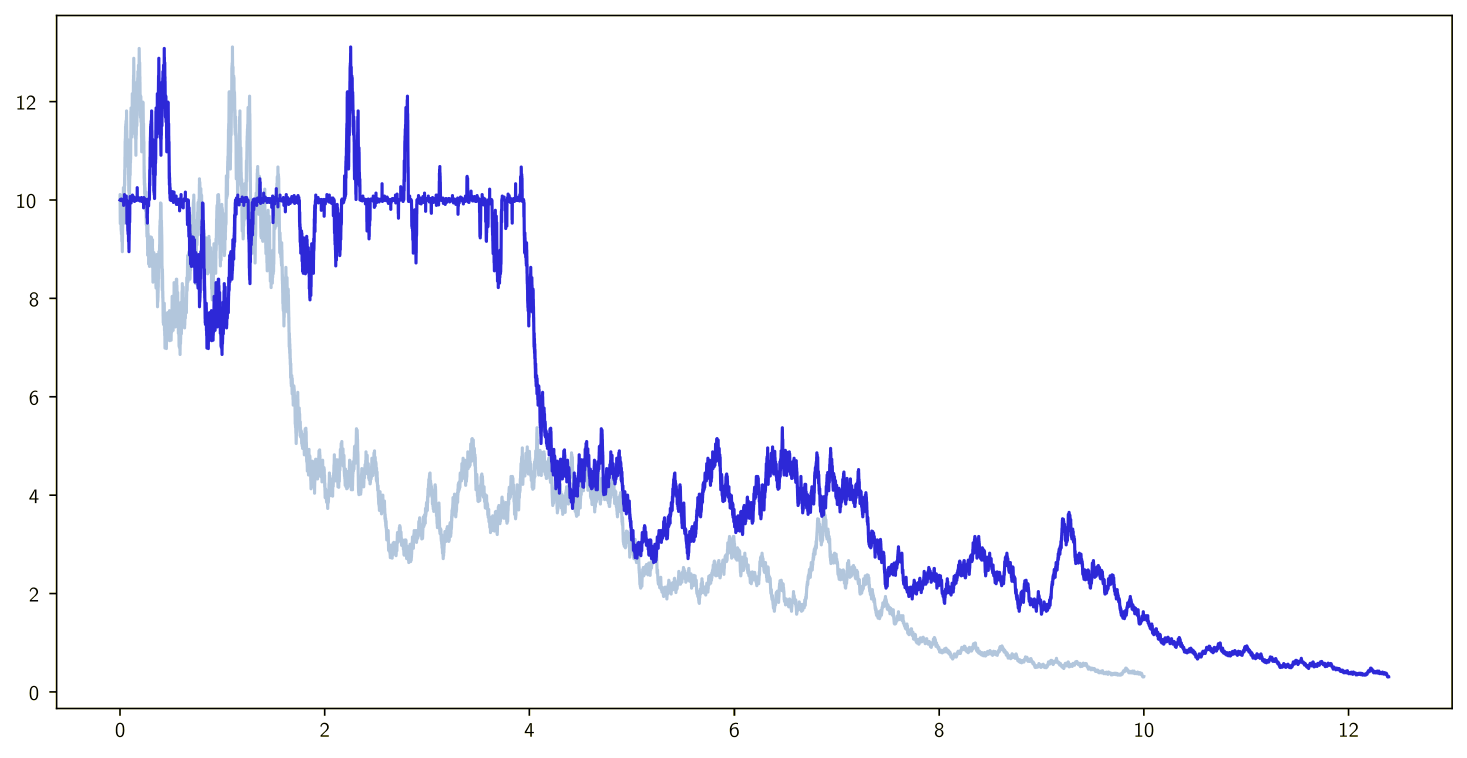}
			\caption{{\centering \small
					Simulated sample paths of the  sticky (dark blue) geometric Brownian motion of parameters $(\mu=0,\sigma=0.25,\rho=1,\zeta=10)$
					and the non-sticky (light blue) geometric Brownian motion of parameters $(\mu=0,\sigma=0.25) $.
					The approximation method used for the simulation is from~\cite{AnaLejVil}.
				}
			}
			\label{fig:sample_path}
		\end{center}
	\end{figure}

	The process $S$ defined this way has the following properties:
	\begin{enumerate}
		
		\item It is a semimartingale with the explicit Doob-Meyer decomposition:
		\begin{equation}
			S_t = \mu \int_{0}^{t}  S_s  \indic{S_s \ne \zeta} \vd s 
			+ \sigma \int_{0}^{t} S_s \indic{S_s \ne \zeta} \vd B_s.
		\end{equation}

		\item \label{item_o_diffusion} It is the regular diffusion on $(0,\infty)$ defined through scale and speed $(s,m) $ (see~\cite[Section~4]{Anagnostakis2022}), where $s,m$ are defined for all $x>0 $ by
		\begin{align}\label{eq_txt_sm_sf}
			s'(x)&= e^{-\frac{2\mu}{\sigma^{2}}\int^{x} \frac{1}{\zeta} \vd \zeta},
			&
			m(\rd x)&= \frac{1}{s'(x)} \frac{1}{\sigma^{2}x^{2}}  \vd x +  \frac{\rho}{s'(\zeta)} \delta_{\zeta}(\rd x).
		\end{align}
		We note that for equation~(4.2) in~\cite{Anagnostakis2022} to be correct, 
		the singular part of the expression should be divided by $s'(0) $.
		The author thanks David Criens and his student, Sebastian Hahn, for identifying and communicating this correction. 
		
		\item It spends a positive amount of time at $\zeta $ upon contact. 
		Since $S$ is a regular diffusion on $(0,\infty) $, this follows from the change of scale~\cite[Proposition~VII.3.5]{RevYor},
		~\cite[Exercise~VI.1.14]{RevYor} and the Markov property.
		This can be observed in the simulated path of the sticky geometric Brownian motion shown in Figure~\ref{fig:sample_path}.

		\item It has an alternative path-wise representation akin to that of the geometric Brownian motion. Indeed, from the Itô formula of~\cite[Lemma 4.5]{Anagnostakis2022}, the process $(X, X^{0}) = (\log S, \log S^{0})$ solves:
		\begin{equation}\label{eq_def_logBS_diffusion}
			\begin{dcases}
				\vd X_t  = \left(\mu - \frac{\sigma^{2}}{2}\right) \indic{X_t \ne \log \zeta} \vd t + \sigma \indic{X_t \ne \log \zeta} \vd B_t, \\
				\indic{X_t = \log \zeta} \vd t = \frac{\rho}{\zeta} \vd \loct{X}{\log \zeta}{t},\\
				\vd X_{t}^{0} = r \vd t.
			\end{dcases}
		\end{equation}
		Consequently, we derive the following expressions for the price dynamics:
		\begin{equation}\label{eq_def_price_dynamics}
			\begin{dcases}
				S_t = S_0\exp \left(\left(\mu - \frac{\sigma^{2}}{2}\right) \int_{0}^{t}\indic{S_t \ne \zeta} \vd t 
				+ \sigma \int_{0}^{t} \indic{S_t \ne \zeta} \vd B_t\right), \\
				\indic{S_t = \zeta} \vd t = \rho \vd \loct{S}{\zeta}{t},\\
				S^{0}_t = S^{0}_{0} \exp \left(r  t\right).
			\end{dcases}
		\end{equation}
	\end{enumerate}

	\begin{example}
		[standard Black--Scholes model]
		\label{example_BS}
		If $\rho=0 $, then $S$ is a geometric Brownian motion and the model $(S,S^{0}) $ is the standard Black--Scholes model. 
		Indeed, in this case, $(s,m) $ defined in~\eqref{eq_txt_sm_sf} are the scale function and speed measure of the geometric Brownian motion (see~\cite[p.132]{BorSal}).
		This can also be seen from~\eqref{eq_def_BS_diffusion}, where if $\rho=0 $, the process $S$ has a geometric Brownian dynamic away from $\zeta $ and the time it  spends time at $\zeta$ is of Lebesgue measure $0$.
	\end{example}
	
	\begin{remark}
		The process $(S, S^{0})$ is a $2$-dimensional sticky diffusion that, when $r \neq 0$, cannot be reduced to a time-change. The stickiness affecting the dynamic of the risky asset's price $S$ slows down its motion on $\zeta $, while the log-value of the zero bond $\log(S^{0})$ steadily increases with rate $r$. 
		This characteristic results in 
		the violation the no-arbitrage hypothesis for the model when $r\not = 0 $.
	\end{remark}

	\section{(No) Arbitrage}
	\label{sec_arbitrage}
	
	In this section we study the model in terms of arbitrage.
	The notion of arbitrage of interest is the free lunch with vanishing risk and the absence of it is called no free lunch with vanishing risk or \NFLVR. 
	The reason we focus on \NFLVR{} is that it is the key notion related to the existence of an equivalent local martingale measure (see~\cite[Theorem~12.1.2]{Delbaen2006}) and the induced arbitrage--free pricing framework.

	First, we recall all the necessary notions related to arbitrage along with
	some useful results.
	Then, we adapt these for the sticky Black-Scholes model.
	In particular, we prove that the model satisfies \NFLVR{} \ifff{}
	$r=0 $.
	In this case, we find an equivalent local martingale measure for the discounted price process $\wt S = S/S^{0} $. 
	
	\subsection{Preliminary notions}
	
	We start by introducing several concepts from the arbitrage theory for semimartingales. 
	Let $X$ be a semimartingale defined on the probability space
	$(\Omega,\process{\bF_t},\Prob)$, where the filtration $\process{\bF_t} $ is assumed both right-continuous and complete.
	This means that
	\begin{enumerate}
		\item for all $t\ge 0 $, $\mc F_t = \bigcap_{s>t} \bF_s $,
		\item $\bF_0 $ contains all $\Prob $-negligible events.
	\end{enumerate}
	
	\begin{SA}
		\label{sa_filtration}
		These assumptions on $\process{\bF_t} $ will apply to all filtrations considered throughout this paper.
	\end{SA}
	
	We may now define the following objects.
	\begin{enumerate}
		\item A \tit{trading strategy} is a predictable process $\theta $ with respect to $\process{\bF_t} $, i.e., for all stopping times $\tau $, $\theta_{\tau} $ is $\bF_{\tau-}$-measurable.
		
		\item An \tit{elementary trading strategy} is a trading strategy of the form:
		\begin{equation}
			\theta_t = \sum_{i=0}^{n} \theta^{(i)} \indic{t\in (\tau_{i},\tau_{i+1}]}
		\end{equation}
		with $n\in \IN $, $(\tau_i)_{i\le n} $ a family of stopping times, and $(\theta^{(i)})_i $ a family of random variables such that, for all $i $, $\theta^{(i)} $ is $\bF_{\tau_i} $-measurable.
		Such strategies are also called simple trading strategies.
		
		\item An $M$-\tit{admissible trading strategy}, for some $M\ge 0 $, is an adapted predictable process $\theta $ with finite variation such that for all $t\ge0 $,  $\int_{0}^{t} \theta_s \vd X_s >-M $ almost surely.
		An \tit{admissible strategy} is a strategy that is $M$-\tit{admissible} for some $M\ge 0 $.
		We denote the set of all $M$-admissible strategies with $\mathcal A_M $.
		
		\item A probability measure measure $\Qrob $ is called an \textit{Equivalent Local Martingale Measure} (or \ELMM{}) for $X$ if $\Qrob \sim \Prob $ and $X$ is a local  martingale on
		$(\Omega,\process{\bF_t},\Qrob) $.
	\end{enumerate}
	
	We now introduce three notions of arbitrage.
	For more on these, see~\cite{Delbaen1994,KaraKard07}.
	We note that arbitrage notions in a diffusion market with non-trivial interest rates should be expressed in terms of the discounted price process.
	Otherwise, for the case $r\in \IR $, depending on the sign of $r $, longing or shorting the riskless asset results in an arbitrage strategy. 
	
	\begin{definition}[(No) Arbitrage]
		A trading strategy $\theta $ is an \textit{Arbitrage Opportunity} for $X$ on $[0, T]$ if $\theta\in \mc A_0$ and $\Prob(\int_{0}^{T} \theta_s \vd X_s > 0) > 0$. 
		A market is \textit{arbitrage-free} (or satisfies \NA) if, for all $\theta \in \mc A_0 $, the probability $\Prob( \int_{0}^{T} \theta_s \vd X_s > 0) = 0$.	
	\end{definition}
	
	\begin{definition}
		[(No) Free Lunch with Vanishing Risk]
		A sequence $(\theta^{n})_n $ of trading strategies is said to realize a \tit{Free Lunch with Vanishing Risk} if for each $n $, $\theta_n $ is $\epsilon_n$-admissible and there exists some 
		$M>0 $ such that,
		\begin{equation}
			\begin{aligned}
				\lim_{n \ra \infty} \epsilon_n &= 0
				& \text{and}&
				& \forall n\in \IN:\quad & \Prob \biggbraces{ \int_{0}^{t} \theta^{n}_s \vd X_s >M} > 0.
			\end{aligned}
		\end{equation}
		A market satisfies No Free Lunch with Vanishing Risk (or satisfies \NFLVR)  
		if there is no sequence of strategies that is free lunch with vanishing risk.
	\end{definition}
	
\begin{definition}
	[(No) Unbounded Profit with Bounded Risk]
	A sequence $(\theta^{n})_n $ is said to realize an \tit{Unbounded Profit with Bounded Risk} if their payoff is uniformly bounded in $(n,t) $ from below but unbounded in probability from above, \ie{} there exists some $M>0 $ such that $(\theta^{n})_n
	\subset \mathcal A_M $ and the payoff sequence is unbounded in probability, \ie  
	\begin{equation}
		\begin{aligned}
		\forall &M'>0, 
		&\exists n&\in \IN, 
		&\exists t&>0:
		& \Prob \biggbraces{\int_{0}^{t} \theta^{n}_s \vd X_s > M'}&>0. 
		\end{aligned}
	\end{equation}
	A market satisfies No Unbounded Profit with Bounded Risk (or satisfies \NUPBR)  
	if there is no sequence of strategies that is an unbounded profit with bounded risk.
\end{definition}

	The reason why we are interested in \NFLVR{} is the following result which qualifies it as a sufficient and necessary condition for the existence of an \ELMM. 
	The latter implies an arbitrage--free pricing framework.
	
	\begin{theorem}[Theorem~12.1.2 of~\cite{Delbaen2006}]
		\label{thm_risk_neutral}
		Let $X$ be a locally bounded semimartingale
		defined on the probability space $(\Omega,\process{\bF_t},\Prob) $.
		The following are equivalent:
		\begin{enumerate}
			\item there exists an \ELMM{} for $X$,
			\item $X$ satisfies \NFLVR,
			\item $X$ satisfies \NA{} and \NUPBR.
		\end{enumerate}
	\end{theorem}
	
	We observe that an arbitrage strategy also realizes a free lunch with vanishing risk.
	This confirms that \NFLVR{} implies \NA. 
	Hence, the following result provides a necessary condition 
	for the existence of an \ELMM.
	
	\begin{theorem}[one-dimensional version of Theorem~12.3.5 of~\cite{Delbaen2006}]\label{thm_NA}
		Let $X$ be a locally bounded semimartingale 
		that satisfies \NA. 
		Then, if $X = M + A$ is the Doob-Meyer decomposition of $X$, then we have that $\vd A = h \vd \qv{M}$, where $h$ is a predictable process.
	\end{theorem}
	
	\subsection{The case of the sticky Black-Scholes model}
	
	We now study the sticky Black--Scholes model in terms of~\NA{} and~\NFLVR.
	We first prove that \NFLVR{} holds when $r =0 $, by finding an \ELMM{} for the discounted price $\wt S = S/S^{0} $.
	We then prove that \NA{}, and hence \NFLVR{}, does not hold when $ r \not=0 $, by establishing an arbitrage strategy.  
	
	We first elaborate on the dynamic of the discounted risky asset.
	From~\eqref{eq_def_BS_diffusion} and the Itô formula~(see \cite[Theorem~IV.3.3]{RevYor}), 
	\begin{equation}
		\label{eq_txt_discounted_dynamic}
		\vd \wt S_t = \vd (e^{-rt}S_t) = e^{-rt} \vd S_t - r e^{-rt} S_t \vd t
		=  \Bigbraces{ \mu \indic{S_t \not = \zeta} -r} \widetilde S_t \vd t
		+ \sigma \widetilde S_t \indic{S_t \not = \zeta} \vd B_t.
	\end{equation} 
	This qualifies $\wt S $ as a semimartingale.
	
	\begin{proposition}\label{prop_riskneutral}
		We consider the sticky Black-Scholes model $(S,S^{0})$, as defined in Section~\ref{sec_model}, defined on the probability space $\mathcal P_x = (\Omega,\process{\bF_t},\Prob_x) $ such that $\Prob_x$--\as, $S_0 = x $.
		If $r= 0 $, then $S= \wt S $ and the following hold.
		\begin{enumerate}
			[label= (\roman*), font=, labelsep=.5em]
			\item \label{item_prop_Girsanov_1} The measure~$\Qrob_x$ defined by 
			$\rd \Qrob_x = \bE(\theta)_T \vd \Prob_x $, where 
			\begin{equation}
				\label{eq_prop_def_riskneutral}
				\begin{aligned}
					\bE(\theta)_T &=\exp \Bigbraces{ \int_{0}^{T}\theta_s \vd B_s
						- \frac{1}{2}\int_{0}^{T} \theta^{2}_s \vd s },
					&\text{and} &
					&\theta_t &= \frac{\mu}{\sigma}, 
					&t &\ge 0,
				\end{aligned}
			\end{equation}	
			is an \ELMM{} for $S$.
			\item \label{item_prop_Girsanov_2} The process $S$ solves
			\begin{equation}\label{eq_prop_risk_neutral_diffusion}
				\begin{dcases}
					\vd S_t   = \sigma S_t \indic{S_t \ne \zeta}
					\vd \wt B_t , \\
					\indic{S_t = \zeta} \vd t = \rho 
					\vd \loct{S}{\zeta}{t}, 
				\end{dcases}
			\end{equation}	
			where $\wt B = \process{B_t - (\mu/\sigma) t }$ is a $\Qrob_x$-standard Brownian motion.
			\item \label{item_prop_Girsanov_3} The process $S$ is a martingale under $\Qrob_x $.
		\end{enumerate}
	\end{proposition}
	
	\begin{proof}
		Since $r=0 $, we have that $\wt S = S $. 
		From the Girsanov theorem of~\cite[Lemma~4.4]{Anagnostakis2022}, $S$ solves 
		\eqref{eq_prop_risk_neutral_diffusion} under $\Qrob_x $. This proves~\ref{item_prop_Girsanov_1}
		and~\ref{item_prop_Girsanov_2}.
		
		The process $S$ is a local martingale under $\Qrob_x $ as the stochastic integral of a predictable process with respect to a standard Brownian motion.
		From~\cite[Ch.IV, Corollary 1.25]{RevYor}, to prove that $S$ is a martingale, it suffices to show that $\Esp^{\Qrob_{x}}\braces{\qv{S}_t} < \infty $.
		Indeed, from \eqref{eq_def_price_dynamics}, we have  
		\begin{equation}
			\begin{aligned}
				\Esp^{\Qrob_{x}} \bigbraces{ \qv{S}_t }
				&= \Esp^{\Qrob_{x}} \Bigbraces{ \int_{0}^{t} \sigma^{2} S^{2}_s \indic{S_s \not = \zeta} \vd s }
				\le \sigma^{2}  \int_{0}^{t} \Esp^{\Qrob_{x}} \bigbraces{ S^{2}_s} \vd s 
				\\
				&= \sigma^{2} S^{2}_0 
				\int_{0}^{t}
				\Esp^{\Qrob_{x}} \Bigbraces{\exp \Bigbraces{ - \sigma^{2} \int_{0}^{s}\indic{S_u \ne \zeta'} \vd u 
						+ 2 \sigma \int_{0}^{s} \indic{S_u \ne \zeta'} \vd B_u}
				} \vd s, 
			\end{aligned}
		\end{equation}
		where from~\cite[Corollary~VIII.1.16]{RevYor}, the quantity under the expectancy is a martingale.
		Hence for all $t\ge 0 $, we have that $\Esp^{\Qrob_{x}} \bigbraces{ \qv{S}_t } \le \sigma^{2} S^{2}_0 t $,
		which is finite.
		This proves~\ref{item_prop_Girsanov_3}, which completes the proof. 
	\end{proof}

	\begin{theorem}\label{thm_arbitrage} 
		The sticky Black-Scholes model satisfies \NA{} and \NFLVR{} if and only if $r =0 $.
	\end{theorem}

	\begin{proof}
		We recall that we work under the Standing Assumption~\ref{sa_filtration}. 
		From Theorem~\ref{thm_risk_neutral} and Proposition~\ref{prop_riskneutral} if $r=0 $ both \NA{} and \NFLVR{} are satisfied.
		From Theorem~\ref{thm_risk_neutral}, it thus suffices to prove the existence of an arbitrage strategy when $r\not = 0 $.
		Inspired by the proof of \cite[Theorem~12.3.5]{Delbaen2006}, we will construct an arbitrage strategy.
		
		We consider the trading strategy $H $ defined by 
		\begin{equation}
			\label{eq_proof_arbitrage_strat}
			H_t = \sgn(-r) e^{rt}\indic{S_t = \zeta}, \qquad t>0.
		\end{equation}
		We now show that for all $t\ge 0$,
		$\int_{0}^{t} H_s \vd \wt S_s \ge 0 $, and that there exists some 
		$t>0 $ such that $\Prob( \int_{0}^{t} H_s \vd \wt S_s >0)>0 $.
		From~\eqref{eq_def_BS_diffusion}, \eqref{eq_txt_discounted_dynamic} and~\eqref{eq_proof_arbitrage_strat}, 
		\begin{equation}
			\begin{aligned}
				\int_{0}^{t} H_s \vd \wt S_s 
				&= \int_{0}^{t} H_s \vd \wt S_s
				= \int_{0}^{t} H_s \indic{S_s=\zeta} \vd \wt S_s
				\\ 
				&
				= - r\zeta \int_{0}^{t} H_s e^{-rs} \indic{S_s=\zeta} \vd s
				= |r| \zeta \int_{0}^{t}  \indic{S_s=\zeta}  \vd s  = |r| \rho  \zeta \loct{S}{0}{t}.
			\end{aligned}
		\end{equation}
		We observe that for all $t\ge 0 $, $ \int_{0}^{t} H_s \vd \wt S_s  \ge 0  $, \ie{} $H\in \mc A_0 $.
		From~\cite[Corollary~29.18]{Kal} and~\cite[Theorem~1.1]{BrugRuf16},  
		$\Prob_x \bigbraces{\int_{0}^{t} H_s \vd \wt S_s  >0} = \Prob_x(\loct{S}{0}{t}>0)= \Prob_x \braces{\inf\{s>0: S_s = \zeta\}<t}>0 $.
		These qualify $H$ as an arbitrage strategy.
		This completes the proof.
	\end{proof}
	
			\begin{example}
		[standard Black--Scholes ($\rho=0 $), continuation of Example~\ref{example_BS}]
		\label{example_BS_arbitrage}
		In the standard Black--Scholes model the discounted price process 
		has dynamic
		\begin{equation}
			\vd \wt S_t  
			=  \bigbraces{ \mu  -r} \wt S_t \vd t
			+ \sigma \wt S_t \vd B_t.
		\end{equation}
		An \ELMM{} for $\wt S $ is the measure $\Qrob_x $ defined by 
		\begin{equation}
			\begin{aligned}
				\vd \Qrob_x &= \bE(\theta)_T \vd \Prob_x,
				& \text{where}&
				&
				\bE(\theta)_T &=\exp \Bigbraces{ \int_{0}^{T}\frac{\mu-r}{\sigma} \vd B_s
					- \frac{1}{2}\int_{0}^{T} \Bigbraces{\frac{\mu-r}{\sigma}}^{2} \vd s }.
			\end{aligned}
		\end{equation}	
		Hence, \NFLVR{} is satisfied.
		
		We note also that in the standard Black--Scholes model, if $A+M $ is the Doob--Meyer decomposition of the discounted price--process $\wt S $, we have that $\vd A_t \ll \vd \qv{M}_t $.
		Indeed, 
		\begin{equation}
			\forall t\ge 0, \qquad
			\begin{dcases}
				A_t = \int_{0}^{t} (\mu  - r) \widetilde S_s   \vd s,
				\\
				M_t = 
				\int_{0}^{t} \sigma \widetilde S_s  \vd B_s.
			\end{dcases}
		\end{equation}
	\end{example}

	\begin{remark}
		\label{remark_stickiness_persistence}
		From Proposition~\ref{prop_riskneutral}, we note that the price process under the \ELMM{} $ \Qrob_x$ is a driftless sticky geometric Brownian motion, sticky at $\zeta>0$, of stickiness parameter $\rho $.
		Therefore, under $\Qrob_x $ the process $S$ is also a diffusion on $(0,\infty) $.
		This also implies that stickiness persists under the \ELMM{} $\Qrob_x $. 
	\end{remark}

	\begin{remark}
		\label{rmk_Q_non_uniqueness}
		In the sticky Black--Scholes model ($\rho>0$), there is no uniqueness of \ELMM{} for $S$. 
		Indeed, we observe that taking any $\theta $ in~\eqref{eq_prop_def_riskneutral} such that 
		\begin{equation}
			\begin{aligned}
				\forall t &\ge 0: 
				&\int_{0}^{t}\theta_s \indic{S_s \not=\zeta} \vd s  &= \int_{0}^{t}
				 \frac{\mu}{\sigma} \indic{S_s \not = \zeta} \vd s, 
			\end{aligned}
		\end{equation}
		also results in an equivalent local martingale measure $\Qrob^{\theta} $ for $S$. 
		In this case, the process $S $ solves
		\begin{equation}
			\begin{dcases}
				\vd S_t  = \sigma S_t  \indic{S_t \ne \zeta}
				\vd  B^{\theta}_t , \\
				\indic{S_t = \zeta} \vd t = \rho 
				\vd \loct{S}{\zeta}{t}, 
			\end{dcases}
		\end{equation}	
		where $\wt B^{\theta} = \process{B_t - \int_{0}^{t}\theta_s \vd s } $ is a $\Qrob^{\theta} $-Brownian motion
	\end{remark}

	\begin{remark}
		The arbitrage strategy constructed in the proof
		of Theorem~\ref{thm_arbitrage} is highly erratic: it is supported on
		$\{t\in[0,T]:\; S_t= \zeta\} $. It is therefore not a simple arbitrage (an arbitrage that is an elementary strategy).
		In the forthcoming paper~\cite{Anagnostakis2024weak}, it is proven that there is no simple arbitrage strategy in this model for any $r\in \IR $.
		This leads to two key implications: First, arbitrage strategies are practically unattainable as they require an infinite number of transactions. Second, by Theorem~\ref{thm_risk_neutral}, the presence of such an unattainable arbitrage precludes the existence of an \ELMM{} for $\wt S $ and of an arbitrage--free pricing framework.
	\end{remark}

	\section{Pricing and hedging}
	\label{sec_pricing}
	
	This section is dedicated to the pricing and hedging problem in the sticky Black--Scholes model.
	We first prove the pricing equation and deduce pricing and hedging expressions.
	These require the existence of a unique classical solution to the pricing problem.
	Next, we apply the martingale representation theorem to prove, under an integrability condition, the existence of a replication strategy and the representation of the claim's price as the expectancy of the payoff under the \ELMM.
	In our analysis, we observe another distinctive feature of the sticky model: the non-uniqueness of the replication strategy.
	
	These derivations are only valid for the case $r=0$. When $r\neq0$, the \NFLVR{} condition is not met and, as per Theorem~\ref{thm_risk_neutral}, there is no \ELMM{} for the discounted price process. 
	In this case, this pricing framework does not apply.
	
	We recall the concept of arbitrage pricing principle.
	
\begin{definition} [see~\cite{Kabanov2009}, Section~2.1.9] \label{def_risk_neutral} 
	The \tit{Arbitrage pricing principle} (or \APP) refers to the concept that he price of a contingent claim 
	is determined in such a way that no party involved can achieve an arbitrage. 
\end{definition}
	
	For the remainder of this section, we consider the sticky Black-Scholes model defined in Section~\ref{sec_model} with $r=0$, defined on the probability space $\mc P_x = (\Omega,\process{\bF_t},\Prob_x) $ so that $S_{0}=x$, $\Prob_x$-almost surely. 
		Let $\Qrob_x$ be the \ELMM{} for $S$ derived in Proposition~\ref{prop_riskneutral}.
	
	We also consider the notation of Markovian families used in~\cite{Fre} for the diffusion $S$ under the equivalent local martingale measure $\Qrob_x $.
		Let $(s_{\Qrob},m_{\Qrob}) $ be the scale and speed of $S$ under $\Qrob_x $.
		We note with $(Q^{S,\Qrob}_x,\; x\in J )$ the canonical diffusion, of scale and speed $(s_{\Qrob},m_{\Qrob}) $, on the path-space $C([0,\infty),\IR)$ and $Y $ be the coordinate process. 
		Hence, for all almost surely finite stopping time $\tau $, we have
	\begin{equation}
			\label{eq_Markovian_family_formalism}
			Q^{S,\Qrob}_{X_{\tau}} \bigbraces{ h(Y_t) } = \Esp^{\Qrob_x} \bigbraces{ h(X_{\tau+t}) \big| \bF_\tau }.
	\end{equation}

	\subsection{Pricing equation}
	
	We consider the generic pricing problem 
	\begin{equation}
		\label{eq_def_BlackScholes_equ}
		\begin{dcases}
			v_{t}  = -\frac{\sigma^{2} x^{2}}{2} v_{xx}
			, \qquad \forall (t,x) \in [0,T)\times (0,\infty),\\
			\frac{\sigma^{2} \zeta^{2}}{2} v_{xx}(t,\zeta-)
			=\frac{\sigma^{2} \zeta^{2}}{2} v_{xx}(t,\zeta+)
			= \frac{1}{2 \rho} \bigbraces{v_{x}(t,\zeta+)-v_{x}(t,\zeta-)},
			\qquad \forall t \in [0,T),\\
			v(T,x) = h(x), \qquad \forall x \in (0,\infty),
		\end{dcases}
	\end{equation}
	
	We now define the space of classical solutions of~\eqref{eq_def_BlackScholes_equ}.
	This is justified from the forthcoming Remark~\ref{rmk_Feller}.
	
	For $I$ an interval of $[0,\infty) $ and $L$ an interval of $(0,\infty) $, we consider the spaces $C^{\Lop}(I)$, $C^{1,\Lop}(I\times L) $, defined by 
	\begin{equation}
		\label{eq_txt_CL}
		C^{\Lop}(L) :=
		\left\{
		\begin{aligned}
			f &\in C(L) \cap C^{2}(L): \\
				&\quad \frac{\sigma^{2} \zeta^{2}}{2} f''(\zeta-)
				= \frac{\sigma^{2} \zeta^{2}}{2} f''(\zeta-)
				= \frac{1}{2 \rho} \bigbraces{f'(\zeta+)-f'(\zeta-)}
		\end{aligned} 
		\right\},
	\end{equation} 
	and 
	\begin{equation}
		C^{1,\Lop}(I \times L)
		:=
				\left\{
		\begin{aligned}
				 u &\in C(I \times L):\; 
				 \\
				&\forall x \in L,\, u(\cdot,x) \in C^{1}(I);\quad  
				 \forall t \in I,\, u(t,\cdot) \in C^{\Lop}(L) 
		\end{aligned} 
		\right\},
	\end{equation}
	where all considered derivatives are right--derivatives.
	
	We can now define the notion of classical solution of~\eqref{eq_def_BlackScholes_equ}.
	
	\begin{definition}
		A function $v$ is said to be a classical solution of~\eqref{eq_def_BlackScholes_equ} if $v\in C^{1,\Lop}([0,T) \times (0,\infty))$ and if it solves~\eqref{eq_def_BlackScholes_equ}. 
	\end{definition}
	
	We note that the space $C^{\Lop}((0,\infty)) $ is comprised of functions that are differences of convex functions.
	Indeed, if $f \in C^{\Lop}((0,\infty)) $, the function $\phi $
	defined by 
	\begin{equation}
		\begin{aligned}
			\phi(x) &= f(x) - \frac{\rho (\zeta\sigma)^{2}}{2} f''(\zeta) (x - \zeta)_+,
			& x&\in (0,\infty),
		\end{aligned}  
	\end{equation}
	is $C^{2}((0,\infty)) $.
	It suffices to observe that the regularity of $\phi $ ensures it is difference of convex functions and that $[x \ra (x - \zeta)_+] $ is convex. 
	This justifies the usage of the Itô-Tanaka formula  (see~\cite[Theorem~VI.1.5]{RevYor}) on any $f\in C^{\Lop}((0,\infty)) $.
	
	\begin{theorem}
		\label{thm_pricing}
		Let $h$ be a measurable function such that~\eqref{eq_def_BlackScholes_equ} 
		has a unique classical solution~$v$.
		Then, the following hold.
		\begin{enumerate}
			[label= (\roman*), font=, labelsep=.5em]
			\item \label{item_pricing} The \APP{} price at $t\in [0,T]$ of the contingent claim with payoff $h(S_T) $ is $v(t,S_t) $.
			\item \label{item_delta} A replication strategy for $h(S_T) $ is
			$\delta = (v_x(t,S_t);\; t\in[0,T]) $.
		\end{enumerate}
	\end{theorem}
	
		\begin{proof}
		Let $v $ be the unique classical solution of~\eqref{eq_def_BlackScholes_equ}.
		Since for all $t\in [0,T) $, $v(t,\cdot)\in  C^{\Lop}((0,\infty)) $, from the Itô-Tanaka formula, 
		\begin{equation}\label{eq_proof_PDE_Ito}
			\begin{aligned}
				\vd v(t,S_t)	
				=& v_t (t,S_t) \vd t + v_x (t,S_t) \vd S_t\\
				&+ \frac{1}{2}  v_{xx} (t,S_t) \vd \qv{ S}_t
				+ \frac{1}{2} \bigbraces{v_{x} (t,\zeta+)- v_{x} (t,\zeta-)} \vd \loct{S}{\zeta}{t}.
			\end{aligned}
		\end{equation}
		From \eqref{eq_prop_risk_neutral_diffusion},
		\begin{equation}
			\begin{aligned}
				\vd \qv{S}_t &= \sigma^{2}S^{2}_t \indic{S_t \not = \zeta} \vd t,
				&
				\indic{S_t = \zeta} \vd t &= \rho  \vd 
				\loct{S}{\zeta}{t}. 
			\end{aligned}
		\end{equation}
		This, combined with~\eqref{eq_def_BlackScholes_equ}, yields
		\begin{equation}
			\frac{1}{2} \bigbraces{v_{x} (t,\zeta+)- v_{x} (t,\zeta-)}
			\vd \loct{S}{\zeta}{t}
			= \frac{\rho \sigma^{2} \zeta^{2}}{4} v_{xx}(t,\zeta)
			\vd \loct{S}{\zeta}{t}
			= - v_t(t,\zeta)\indic{S_t = \zeta} \vd t .
		\end{equation}
		Consequently, 
		\begin{equation}\label{eq_proof_PDE_Itob}
			\begin{aligned}
				\vd v(t,S_t) = & 
				v_t (t,S_t) \indic{S_t \not = \zeta} \vd t + v_x (t,S_t) \vd S_t 
				+ \frac{\sigma^{2} S^{2}_t}{2}  v_{xx} (t,S_t)  \indic{S_t \not = \zeta} \vd t\\
				&+ v_t (t,S_t) \indic{S_t = \zeta} \vd t
				+ \frac{1}{2} \bigbraces{v_{x} (t,\zeta+)- v_{x} (t,\zeta-)} \vd \loct{S}{\zeta}{t}\\
				=& v_x (t,S_t) \vd S_t + 
				\Bigbraces{v_t (t,S_t) 
					+ \frac{\sigma^{2} S^{2}_t}{2}  v_{xx} (t,S_t) } 
				\indic{S_t \not = \zeta}  \vd t.
			\end{aligned}
		\end{equation}
		Integrating in time yields 
		\begin{equation}
			h(S_T) 
			=
			v(t,S_t) + 
			\int_{t}^{T} v_x (t,S_s) \vd S_s.
		\end{equation}
		This proves the two assertions. 
	\end{proof}
	
	\begin{remark}
		\label{rmk_discontinuity}
		The replication strategy $\delta = (v_x(t,S_t);\; t\in[0,T]) $ in the sticky Black-Scholes model is typically discontinuous at the sticky threshold $\zeta $.
		Indeed, from~\eqref{eq_def_BlackScholes_equ}, 
		if $v_{xx}(t,\zeta)\not=0$, then 
		$v_x(t,\zeta+) - v_x(t,\zeta-) \not = 0 $.
	\end{remark}
	
	For our next remark, we recall some notions. 
	The semi--group of a diffusion is the family of operators $\process{P_t} $
	defined by
	\begin{equation}
		\begin{aligned}
			P_t f(x) &:= \Esp \bigbraces{f(S_t) \big| S_0 = x}, 
			& \forall&\; f \in C_b((0,\infty)),\; t\ge 0,\; x\in (0,\infty).
		\end{aligned}
	\end{equation}
	We note that $S$ is a Feller process (see~\cite[Theorem~1]{criens2024FellerDynkin}), \ie{} it is strong Markov and its semi--group $\process{P_t} $ is internal on $C_0((0,\infty)) $ and continuous at $0$, \ie{} 
	\begin{enumerate}
		\item $\forall f \in C_0((0,\infty)) $: $P_t f \in C_0((0,\infty)) $,
		\item $\forall f \in C_0((0,\infty)),\; x\in (0,\infty) $: $\lim_{t\ra 0} P_t f(x) = f(x) $.
	\end{enumerate} 
	Its semi--group is therefore generated by the following operator 
	\begin{equation}\label{eq_proof_infinitesimal_generator}
		\begin{aligned}
		\Lop f &= \frac{(x\sigma)^{2}}{2} f'',
		& \dom_{C_0}(\Lop) = & C_0((0,\infty)) \cap C^{\Lop}((0,\infty)). 
		\end{aligned}
	\end{equation}

	\begin{remark}
		\label{rmk_Feller}
		We note that if $h \in \dom_{C_0}(\Lop) $, from the Kolmogorov equations (see~\cite[Proposition~VII.1.2]{RevYor}), the problem~\eqref{eq_def_BlackScholes_equ} has a unique classical solution 
		 $v \in C^{1,\Lop}([0,T] \times (0,\infty)) $, that admits the following representation
		\begin{equation}
			\begin{aligned}
				\forall (t,x) &\in [0,T] \times (0,\infty), &v(t,x) &= Q_{x}^{S,\Qrob}(h(Y_{T-t})).
			\end{aligned}
		\end{equation}
	\end{remark}
	
	\subsection{Representation}
	
	In Remark~\ref{rmk_Feller}, we observed that if $h\in \dom_{C_0}(\Lop)$, then
	the pricing equation has a unique classical solution, the price function
	has a representation in terms of an expectancy and there exists a replication strategy.
	In this section we prove that the two latter assertions are part of a more general principle, which is consequence of the martingale representation property~\cite[Theorem~2.7]{criens2024representationproperty1dgeneral}. 
	
	Our motivation is that the condition $h\in \dom_{C_0}(\Lop) $ is quite restrictive. Indeed, classical claims like Calls and Puts have payoff functions that
	are neither in $C_0((0,\infty)) $ nor in $C^{\Lop}((0,\infty)) $.
	This holds true even for the standard Black--Scholes model, where 
	$\dom_{C_0}(\Lop) = C_0((0,\infty)) \cap C^{2}((0,\infty)) $.
	Nevertheless, as depicted in Section~\ref{ssec_discrete_time_hedging}, we obtain satisfactory numerical results for a Call payoff in the sticky and the non-sticky model.

	\begin{theorem}
		\label{thm_prepricing}
		For all $h \in \msf H(S_T) := \{ g\; \text{\normalfont measurable}:\; g(S_T) \in L^{1}(\Omega,\bF_T,\Qrob_x) \}$, 
		\begin{enumerate}
			[label= (\roman*), font=, labelsep=.5em]
			\item \label{item_prereplication} there exists a trading strategy $ \delta $ in $ L^{2}_{\loc}(S)$
			(for a definition of $ L^{2}_{\loc}(S)$, see~\cite[\S I.4.39]{JacodShiryaev2003_limit}), 
			such that for all $t\in [0,T] $, 
			\begin{equation}
				\label{eq_thm_replication}
				\begin{dcases}
				\Esp^{\Qrob_x}(h(S_T) | \bF_t) = \Esp^{\Qrob_x}(h(S_T)) + \int_{0}^{t} \delta_s \vd  S_s,
				\\
				h(S_T) = \Esp^{\Qrob_x}(h(S_T) | \bF_t) + \int_{t}^{T} \delta_s \vd  S_s,  
				\end{dcases}
			\end{equation}
			\item \label{item_prepricing} and for all $t\in [0,T] $, the \APP{} price of the claim of terminal payoff $h(S_T) $ at $t$ is $V_t := \Esp^{\Qrob_x}(h(S_T) | \bF_t) = Q^{S,\Qrob}_{S_t} \bigbraces{ h(Y_{T-t}) }$.
		\end{enumerate}
	\end{theorem}
	
	\begin{proof}
		Regarding~\ref{item_prereplication}, since $h(S_T) \in L^{1}(\Omega,\bF_T,\Qrob_x) $, the process $M$ defined as follows is a uniformly integrable martingale 
		\begin{equation}
			M_t := \Esp^{\Qrob_x}(h(S_T) | \bF_t), \qquad t\in [0,T].
		\end{equation}
		From the Markov property of $S$ under $\Qrob_x $ (see Remark~\ref{remark_stickiness_persistence}), we have that $\Esp^{\Qrob}(h(S_T) | \bF_t)
		= \Esp^{\Qrob}(h(S_T) | \sigma(S_t)) $. Hence, $M$ is adapted to the canonical filtration of $S$.
		Since $S $ is a local martingale, from~\cite[Theorem~2.7]{criens2024representationproperty1dgeneral},
		there exists a predictable process $\delta \in L^{2}_{\loc}(S) $
		such that 
		\begin{equation}
			M_t = M_0 + \int_{0}^{t} \delta_s \vd  S_s.
		\end{equation}
		This proves the first equation of~\ref{item_prereplication}
		Regarding the second equation, we consider \eqref{eq_thm_replication} for different values of $t $. Substracting them
		results in
		\begin{equation}
			h(S_T) = \Esp^{\Qrob_x}(h(S_T) | \bF_{t_0}) - \int_{t_0}^{T} \delta_s \vd  S_s.
		\end{equation}
		
		Regarding~\ref{item_prepricing}, from the \APP{} (see Definition~\ref{def_risk_neutral}) and the payoff replication relations~\ref{eq_thm_replication}, the price of the claim at $t$ is 
		necessarily $ \Esp^{\Qrob_x}(h(S_T) | \bF_t) $.
		Any other value would result in an arbitrage opportunity for one of the two parties.
		This finishes the proof.
	\end{proof}
	
	\begin{remark}
		We note that in the sticky Black--Scholes model ($\rho>0 $) we have no uniqueness of the payoff replication strategy.
		Indeed, let $\delta $ be a payoff replication strategy
		and $\delta' $ a strategy such that
		\begin{equation}
			\begin{aligned}
				\delta  &\not = \delta' 
				&\text{and}&
				& \forall t&\ge 0: 
				&\delta_t \indic{S_t\not =\zeta}  &= \delta'_t \indic{S_t\not =\zeta}. 
			\end{aligned}
		\end{equation}
		Then, we have  
		\begin{equation}
			\Esp^{\Qrob_x}(h(S_T) | \bF_t) - \Esp^{\Qrob_x}(h(S_T)) 
			=  \int_{0}^{t} \delta_s \vd  S_s
			=  \int_{0}^{t} \delta'_s \vd  S_s.
		\end{equation}
	\end{remark}
	
	\begin{corollary}
		Let $h \in \msf H(S_T)$ such that~\eqref{eq_def_BlackScholes_equ}
		has a unique classical solution $v$.
		Then, from the \APP,
		\begin{equation}
			\begin{aligned}
				\forall (t,x)&\in [0,T] \times (0,\infty) , 
				&v(t,x) &= Q^{S,\Qrob}_{x} \bigbraces{ h(Y_{T-t}) },
			\end{aligned}
		\end{equation}
		and if $\delta $ is a replication strategy for $h(S_T) $, then 
		\begin{equation}
			\begin{aligned}
			\forall t&\in [0,T],
			& v_x(t,S_t) \indic{S_t \not = \delta} &=  \delta_t \indic{S_t \not = \delta},
			\end{aligned}
		\end{equation}
		and $(v_x(t,S_t);\; t\in [0,T] ) \in L^{2}_{\loc}(S) $.
	\end{corollary}
	
	\begin{example}
	[standard Black-Scholes ($\rho=0 $), continuation of Examples~\ref{example_BS},~\ref{example_BS_arbitrage}]
		\label{example_BS_pricing}
		For the standard Black-Scholes model $(\rho = 0)$,
		the risky asset and discounted risky asset price dynamics $S,\wt S$ under $\Qrob_x $ are 
		\begin{equation}
			\label{eq_BS_SDE_12}
			\begin{aligned}
				\vd S_t &= S_t \bigbraces{ r \vd t + \sigma \vd B_t},
				&
				\vd \wt S_t &=  \sigma \wt S_t \vd B_t. 
			\end{aligned}
		\end{equation} 
		The pricing equation is (see~\cite[Theorems~6.4.3 and~7.3.1]{shreve2004stochasticfinance})
		\begin{equation}
			\label{eq_bS_pricing_PDE}
			\begin{dcases}
				v_t(t,x) + rx v_x(t,x) + \frac{\sigma^{2}x^{2}}{2} v_{xx}(t,x)
				= rv(t,x), & \forall (t,x) \in (0,T)\times (0,\infty),\\
				v(T,x) = h(x),
			\end{dcases}
		\end{equation}
		If $h\in \msf H(S_T) $, there exists a replication strategy in $L^{2} $ and the price at $t\in [0,T]$ of the contingent claim of payoff at maturity 
		$h(S_T) $ is given by 
		\begin{equation}
			V_t := \Esp^{\Qrob_x}\bigbraces{ e^{-r(T-t)} h(S_T) \big| \bF_t }.
		\end{equation} 
		If~\eqref{eq_bS_pricing_PDE} has a unique classical solution $v $, the \APP{} price of the contingent claim at $t\in [0,T] $ is $V_t = v(t,S_t) $, and the replication strategy is the process $\delta $ defined by (see~\cite[Remark~7.3.3]{shreve2004stochasticfinance})
		\begin{equation}
			\begin{aligned}
				\delta_t &:= v_x(t,S_t),
				& t&\in [0,T].
			\end{aligned}
		\end{equation}
	\end{example}

	\section{Price monotonicity and payoff inclusion}
	\label{sec_monotonicity}
	
	In this section we prove price monotonicity, 
	and inclusion properties for a class of replicable payoffs, with respect to the stickiness parameter.
	
	We first prove that higher stickiness leads to lower prices for European derivative with convex payoff functions.
	We deduce the respective inclusion on the subset of convex positive replicable payoffs.
	Second, we generalize the latter and prove that within the class of locally bounded 
	($L^{\infty}_{\loc} $) payoffs, the same inclusion holds.
	The proof relies on fine properties of the transition kernel.
	These inclusions justify the usage of classical payoffs in the sticky Black-Scholes model.

	We first define the following objects.
	For each $\rho\ge 0$, let $S^{(\rho)}$ be the solution of \eqref{eq_def_BS_diffusion} and 
	$\mathsf H(S^{(\rho)}_T) $ be the class of replicable payoff functions on $S^{(\rho)}_T $, given by
	\begin{equation}
		\label{eq_H_rho_replicable_payoffs}
		\mathsf H(S^{(\rho)}_T) := \bigcubraces{h \; \text{measurable}:\; h(S^{(\rho)}_T) \in L^{1}(\Omega,\bF_T,\Qrob_x)}.
	\end{equation}
	
	\begin{proposition}[price monotonicity]\label{cor_price_bounds_stickiness}
		Let $\Opt_h(S^{(\rho)},x,T)$ be the \APP{} price at $t=0 $ of the contingent claim on the underlying asset $S^{(\rho)}$ with payoff function $h$, maturity $T>0$ and spot price $S^{(\rho)}_0 = x $, \ie{} 
		$\Opt_h(S^{(\rho)},x,T) = \Esp^{\Qrob_x} \bigbraces{h(S^{(\rho)}_T)}$.
		If $h$ is convex, then 
		\begin{enumerate}
			[label= (\roman *),font=,labelsep=.5em]
			\item \label{item_monotonicity_stickiness} for all $x,T>0 $, the application $[\rho \rightarrow \Opt_h(S^{(\rho)},x,T)] $ is decreasing,
			\item \label{item_monotonicity_time} for all $x>0 $, $\rho\ge0 $, the application  $[T \rightarrow \Opt_h(S^{(\rho)},x,T)] $ is increasing,
		\end{enumerate}
	\end{proposition}
	
	The proof relies on the following result, which allows a sticky diffusion to be expressed as a time change of a non-sticky diffusion. Although this result can be inferred from the proof of~\cite[Proposition~2.1.17]{Anagnostakis2022Thesis}, we provide a more detailed proof in Appendix~\ref{app_lem_timechange}.
	
	\begin{lemma}
		\label{lem_timechange}
		Let $X$ be a regular diffusion on $J $, an open interval of $\IR $, on natural scale of speed measure $m$ so that 
		$\rho := m(\{\zeta\})\in (0,\infty) $.
		We suppose $X$ to be defined on the probability space $\mc P_x = (\Omega,\process{\bF_t},\Prob_x) $.
		There exists a regular diffusion $Z $ on $ J$, on natural scale, with speed measure $m_Z = m - \rho \delta_{\zeta} $, defined on an extension of $\mc P_x $,  such that 
		\begin{equation}
			\forall t\ge 0:\qquad
			X_{t} = Z_{\gamma_{\rho}(t)},
			\qquad A_{\rho}(t) = t + \rho \loct{Z}{\zeta}{t},
		\end{equation}
		where $\gamma_{\rho} $ is the right-inverse of $A_{\rho} $.
	\end{lemma}
	
	\begin{proof}
		[Proof of Proposition~\ref{cor_price_bounds_stickiness}]
		Let $S^{(0)}$ be defined on the probability space 
		$(\Omega,\process{\bF_t},\Prob_{x}) $ such that 
		$ S^{(0)}=x$, $\Prob_x$-almost surely. 
		Let $ \Qrob_x$ be the \ELMM{} for $S^{(0)}$ 
		defined in Proposition~\ref{prop_riskneutral}.
		From Theorem~\ref{thm_prepricing}\ref{item_prepricing}, 
		\begin{align}	
			\Opt_h(S^{(0)},x,T) &= \Esp^{\Qrob_x} (h(S^{(0)}_{T})).
		\end{align}
		For all $\rho \ge 0 $, let $A_{\rho} $ be the random time change defined by $A_{\rho}(t)= t + \rho \loct{S^{(0)}}{\zeta}{t} $ and let $\gamma_{\rho} $ be its right-inverse.
		We observe that 
		\begin{enumerate}
			\item if $\rho'<\rho $, then $0 \le \gamma_{\rho}(T)\le \gamma_{\rho'}(T)\le T $,
			\item from Lemma~\ref{lem_timechange}, $\process{S^{(0)}_{\gamma_{\rho}(t)}}=\process{S^{(\rho)}_t} $ in law,
			\item $\gamma_{\rho}^{-1},\gamma_{\rho} $ are stopping times,
			\item since $S^{(0)} $ is a $\Qrob_x $-martingale (Proposition~\ref{prop_riskneutral}\ref{item_prop_Girsanov_3}), $\process{h(S^{(0)}_t)} $ is a $\Qrob_x $-submartingale.
		\end{enumerate}
		Thus, from the stopping theorem, if $\rho'<\rho $,
		\begin{equation}
			\begin{aligned}
				\Opt_h(S^{(\rho')},x,T) 
				&= \Esp^{\Qrob_x} (h(S^{(\rho')}_{T}))
				= \Esp^{\Qrob_x} (h(S^{(0)}_{\gamma_{\rho'}(T)}))
				\\ &\le \Esp^{\Qrob_x} (h(S^{(0)}_{\gamma_{\rho}(T)}))
				= \Esp^{\Qrob}_{x} (h(S^{(\rho)}_{T}))
				= \Opt_h(S^{(\rho)},x,T).
			\end{aligned}
		\end{equation}
		This proves~\ref{item_monotonicity_stickiness}.
		With similar arguments, if $0 \le T\le T' $,
		\begin{equation}
			\Opt_h(S^{(\rho)},x,T) 
			= \Esp^{\Qrob_x} (h(S^{(\rho)}_{T}))
			\le \Esp^{\Qrob_x} (h(S^{(\rho)}_{T'}))
			= \Opt_h(S^{(\rho)},x,T').
		\end{equation}
		This proves~\ref{item_monotonicity_time}, which completes the proof.
	\end{proof}
	
	\begin{remark}
		Lemma~\ref{lem_timechange} yields that if $C^{+} $ is the cone of positive convex
		functions, then 
		\begin{equation}
			\label{eq_txt_rho_convex_inclusion}
			\begin{aligned}
			\msf H(S^{(\rho)}_T) &\cap C^{+}
			\subset \msf H(S^{(\rho')}_T) \cap C^{+},
			& \forall\, &0 \le \rho < \rho'.
			\end{aligned}
		\end{equation}
		In particular, this holds for Call and Put payoff functions.
	\end{remark} 
	
	The next proposition goes beyond this observation.
	
	\begin{proposition}
		\label{cor_replicable_payoff_class}
		Let $ \Lloc $ be the family of locally bounded functions on $(0,\infty) $, \ie{} 
		the functions that are bounded on all compacts of $(0,\infty)$.
		For all $0\le \rho<\rho' $, we have
		\begin{equation}
			\msf H(S^{(\rho)}_T) \cap \Lloc
			\subset \msf H(S^{(\rho')}_T) \cap \Lloc.
		\end{equation}  
	\end{proposition}
	
	For this, we use the following lemma.
	For a proof, we refer the reader to Appendix~\ref{app_lem_kernel_monotony}.
	
	\begin{lemma}
		\label{lem_kernel_monotony}
		Let $[(t,x,y)\ra p_0(t,x,y)] $ and $m_0(\rd y) = \dot m_0( y) \vd y $ be the probability transition kernel and speed measure of the  geometric Brownian motion.
		For all $x>0 $ and $T>0 $, there exists a compact $K_{x,T}$ such that on $(t,y)\in [0,T] \times K^{c}_{x,T} $, $[t\ra p_0(t,x,y) \dot m_0(y)] $ is increasing.
	\end{lemma}
	
	\begin{proof}
		[Proof of Proposition~\ref{cor_replicable_payoff_class}]
		For all $\rho\ge 0 $, let $m_{\rho} $ be the speed measure, and $[(t,x,y) \ra p_{\rho}(t,x,y)] $ the probability transition kernel, of the process $S^{(\rho)} $. 
		Let $\gamma_{\rho} $ be the right--inverse of the time--change $A_{\rho}(t) = t + \rho \loct{S}{(0)}{t} $, $t\ge 0 $. 
		This yields in particular that for all $t\ge 0 $, $\gamma_{\rho}(t)\le t $.
		
		From Lemma~\ref{lem_timechange}, since $S^{(\rho)} = \process{S^{(0)}_{\gamma_{\rho}(t)}} $ in law, for all positive measurable
		function $f$, 
		\begin{equation}
			\label{eq_proof_exp_timechange}
			\Esp^{\Qrob_x} \bigbraces{f(S^{(\rho)}_T)}
			=
			\Esp^{\Qrob_x} \bigbraces{f(S^{(0)}_{\gamma_{\rho}(T)})}
			=
			\Esp^{\Qrob_x} \biggbraces{\int_{(0,\infty)} f(y) p_0 (\gamma_{\rho}(T),x,y) \dot m_{0}( y) \vd y}. 
		\end{equation}
 		Let $K_{x,T} $ be some compact set of $(0,\infty) $ inferred from Lemma~\ref{lem_kernel_monotony}.
 		Hence,  from~\eqref{eq_proof_exp_timechange}, 
		\begin{equation}
			\label{eq_proof_kernek_KKc_decomp}
			\begin{aligned}
				\Esp^{\Qrob_x} \bigbraces{|h(S^{(\rho)}_T|)}
				=& 
				\Esp^{\Qrob_x} \bigbraces{\indic{S^{(\rho)}_T \in K_{x,T}}|h(S^{(\rho)}_T)|}
				+ \Esp^{\Qrob_x} \bigbraces{\indic{S^{(\rho)}_T \in K^{c}_{x,T}}|h(S^{(\rho)}_T)|}
				\\ \le& \xnorm{\indic{K_{x,T}} h}{\infty}
				+ \Esp^{\Qrob_x} \bigbraces{\indic{S^{(0)}_{\gamma_{\rho}(T)} \in K^{c}_{x,T}}|h(S^{(0)}_{\gamma_{\rho}(T)})|}
				\\ \le & \xnorm{\indic{K_{x,T}} h}{\infty} + \Esp^{\Qrob_x} \biggbraces{\int_{0}^{\infty} \indic{y \in K^{c}_{x,T}} |h(y)| p_0 (\gamma_{\rho}(T),x,y) \dot m_{0}( y) \vd y}.		
			\end{aligned}
		\end{equation} 
		Let $0\le \rho < \rho' $. 
		We recall that  for all $T\ge 0$ and $t\in [0,T] $, $0 \le \gamma_{\rho'}(t)\le \gamma_{\rho}(t) \le T$.
		Hence, from the definition of $K_{x,T} $, 
		\begin{equation}
			\Esp^{\Qrob_x} \biggbraces{\int_{K^{c}_{x,T}} |h(y)| p_0 (\gamma_{\rho'}(T),x,y) \dot m_{0}( y) \vd y}
			\le 
			\Esp^{\Qrob_x} \biggbraces{\int_{K^{c}_{x,T}} |h(y)| p_0 (\gamma_{\rho}(T),x,y) \dot m_{0}( y) \vd y}.
		\end{equation}
		Since $h$ is locally bounded, $\xnorm{\indic{K_{x,T}} h}{\infty} < \infty $.
		These prove the inclusion  $\msf H(S^{(\rho)}_T) \cap \Lloc
		\subset \msf H(S^{(\rho')}_T) \cap \Lloc$. 
		This completes the proof.
	\end{proof}

	\begin{remark}
		Proposition~\ref{cor_replicable_payoff_class} 
		ensures that the $L^{1}$-condition of Theorem~\ref{thm_prepricing} is met for all locally bounded payoff functions that are replicable in the standard Black-Scholes model, \ie{} for all $h \in \mathsf H(S^{(0)}_T) \cap \Lloc $.
	\end{remark}

	\begin{figure}[t!]
		\begin{center}
			\includegraphics[alt={Price curves for sticky and non-sticky Bloack-Scholes},width = 0.48\textwidth]{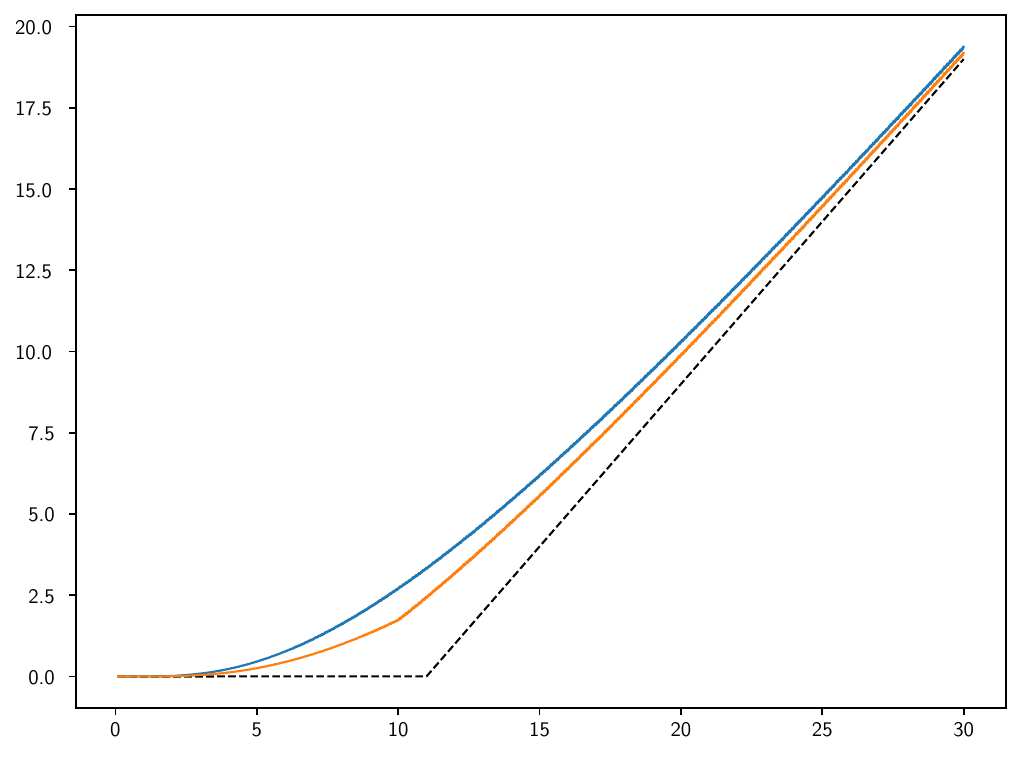}~
			\includegraphics[alt={Delta curves for sticky and non-sticky Bloack-Scholes},width = 0.475\textwidth]{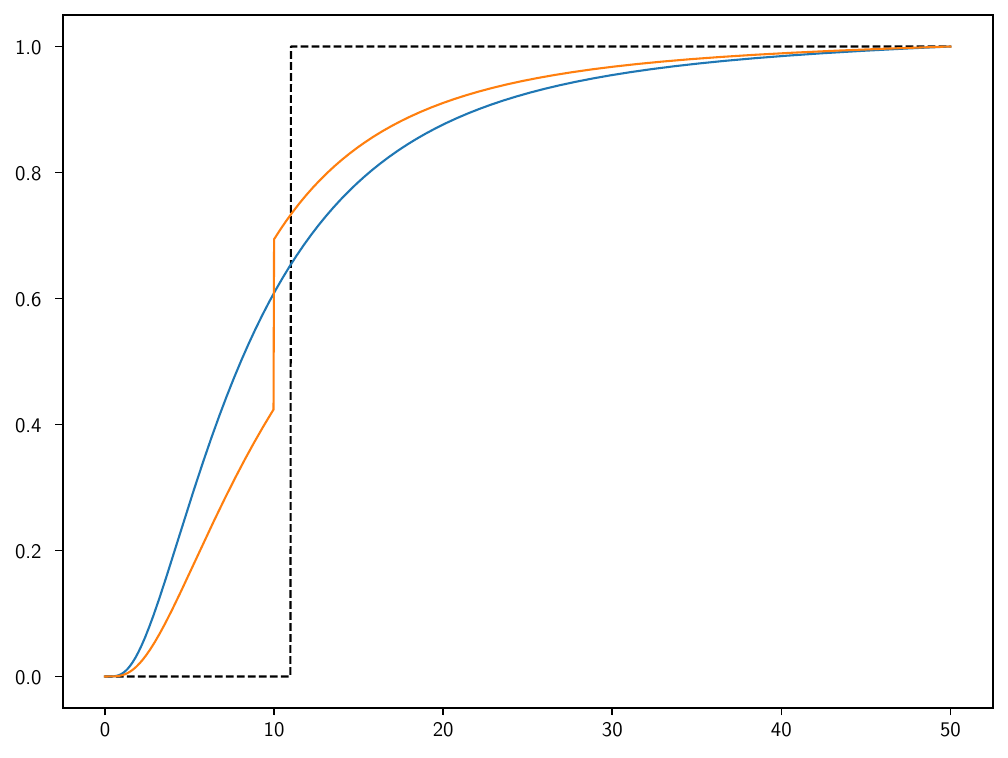}\\
			\includegraphics[alt={Price curves for sticky and non-sticky Bloack-Scholes},width = 0.48\textwidth]{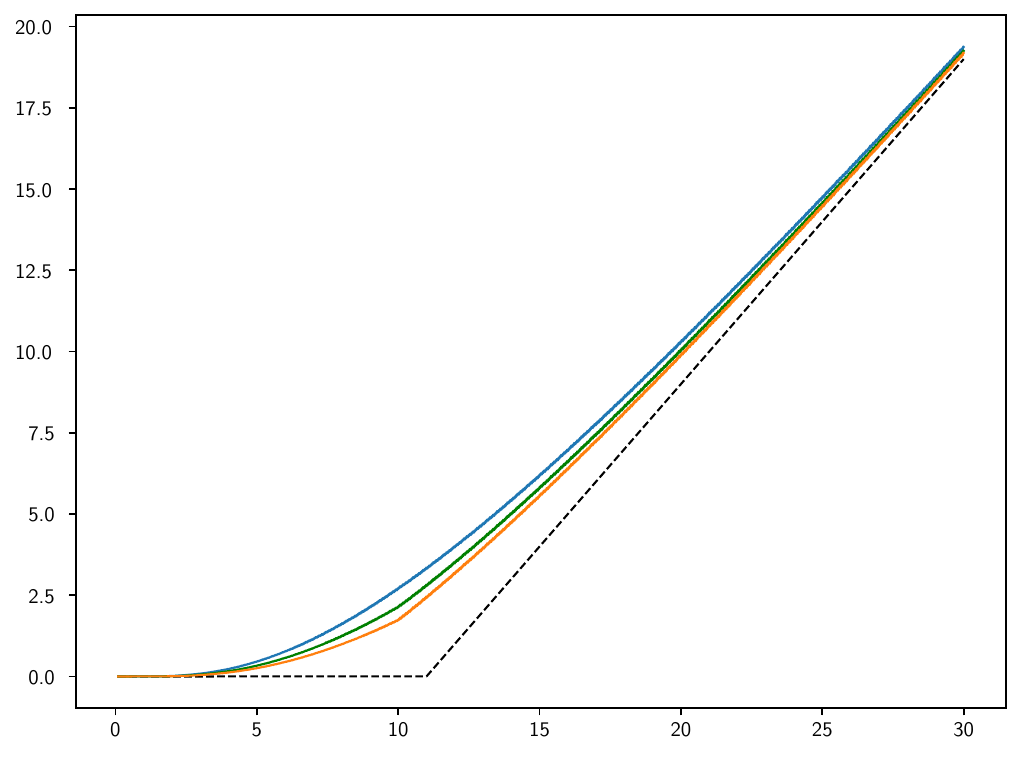}~
			\includegraphics[alt={Delta curves for sticky and non-sticky Bloack-Scholes},width = 0.475\textwidth]{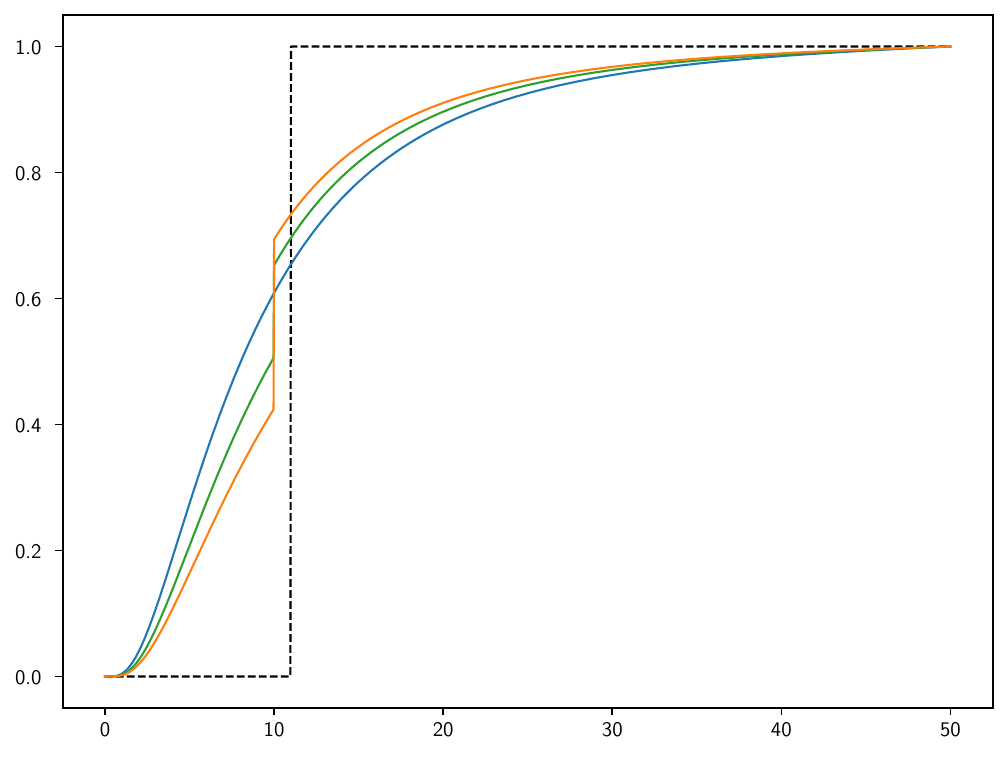}
			\caption{
				\small
				\centering
				\tbf{Price and delta of the Black-Scholes and sticky Black-Scholes.}
				Finite-elements approximations of the price (left) and delta (right) curves of a Call option of strike $K=11 $ and $ T=10$ in the $(r,\sigma)=(0,0.25) $ standard Black-Scholes model (blue), the $(r,\sigma,\rho,\zeta)=(0,0.25,1.0,10.0) $ sticky Black-Scholes model (green)
				and the $(r,\sigma,\rho,\zeta)=(0,0.25,2.0,10.0) $ sticky Black-Scholes model (orange). 
			}
			\label{figure_price_curves2}
		\end{center}
	\end{figure}

	\section{Numerical Experiments}
	\label{sec_numerical}

	This section is dedicated to numerical experiments aimed at evaluating practical aspects of hedging a contingent claim on an asset with a sticky price dynamic. We investigate two key properties. 
	The first property involves the hedging time granularity of the sticky Black-Scholes model, referred to as ``the impact on replication when the agent hedges in discrete time." 
	The second property is model mismatch, or ``the impact on replication when an agent ignores or misrepresents price stickiness."

	In the rest of this section we consider the assets dynamic~\eqref{eq_def_BS_diffusion} with $r=0 $, $\sigma=0.25 $ and various degrees of stickiness~$\rho\in\{0,1,2\}$ at the threshold $\zeta=10 $. 
	We recall that a stickiness value of $\rho=0 $ corresponds to the standard Black-Scholes dynamic. 
	The claim of interest will be a Call option with strike $K$ that delivers $(S_{T}-K)_{+} $ at maturity $T=10 $.
	
	The approximation scheme we use to simulate sample paths of the sticky geometric Brownian motion is the Space--Time Markov Chain Approximation or STMCA of one--dimensional
	diffusions~\cite{AnaLejVil}.
	
	\begin{algorithm}[t!]
		\caption{Delta Hedging of an Option}\label{algo_delta_hedging_sticky}
		\KwData{Stock price time-series $\process{S_t}$, Strike price $K$, Time to maturity $T$, Volatility $\sigma$, Stickiness $\rho$, Hedging window $\Delta t$}
		\KwResult{Deltas and Portfolio Values over time}
		\medskip
		Compute or approximate the price function \texttt{u\_price}$(x,K,T,\sigma,\rho)$\;
		Compute or approximate the delta function \texttt{u\_delta}$(x,K,T,\sigma,\rho)$\;
		\medskip
		Initialize $cash$ to \texttt{u\_price}$(S_0,K,T,\sigma,\rho)$\;
		Initialize $t$ to $0$\;
		Initialize $old\_delta$ to $0$\;
		\medskip
		\While{$t < T$}{
			Update time: $t \leftarrow t + \Delta t$\;
			Retrieve current stock price $S_t$\;
			Calculate current option value: $price \leftarrow$ \texttt{u\_price}$(S_t, K, T - t, \sigma, \rho)$\;
			Calculate current option delta: $delta \leftarrow$ \texttt{u\_delta}$(S_t, K, T - t, \sigma, \rho)$\;
			Update cash position: $cash \leftarrow cash - (delta - old\_delta) \times S_t$\;
			Set $old\_delta \leftarrow delta$\;
			Update portfolio value: $portfolio\_value \leftarrow cash + (delta \times S_t) - price$\;
			Append $delta$ to the list of deltas\;
			Append $portfolio\_value$ to the list of portfolio values\;
		}
		\medskip
		\Return{List of deltas, List of portfolio values}\;
	\end{algorithm}

	\subsection{Hedging time-granularity}\label{ssec_discrete_time_hedging}
	
	The message of the Black-Scholes model is as follows:
	It is possible to replicate the payoff of a European option using a self-financing portfolio composed of quantities of the risky and non-risky assets. In particular, the option seller can hedge their risk by holding a portfolio that, at each time $t$, comprises $\delta_t = v_x(t,S_t) $ parts of the risky asset.
	
	However, a practical challenge arises since achieving this requires continuous rebalancing of the portfolio composition over time. This, as mentioned in~\cite{Bertsimas2001}, proves impossible in real-world scenarios as it necessitates an infinite number of operations. One way to address this is to introduce a hedging window $h>0$, resulting in hedging times $(t_{i})_i $ such that $t_{i+1}-t_{i} = h $. The agent can then rebalance the portfolio at each time $t_i $ so that for all $i $, they hold $v_{x}(t_i,S_{t_i}) $ parts of the risky asset on $[t_i,t_{i+1}) $. 
	The resulting discretized version of the continuous-time hedging strategy $\delta $ is the strategy $\delta^{(N)} $
	defined for all $t $ by 
	\begin{equation}
		\delta^{(N)}_t = \delta_{[Nt]/N},
	\end{equation}
	where $[x]=\sup\{k\in \IN: k \le x\} $.
	
	The PnLs at $t$ of the strategies $\delta $ and $ \delta^{(N)}$ are 
	\begin{align}
		V_t(\delta) - V_0(\delta) &= \int_{0}^{t} \delta_s \vd S_s,
		&
		V_t(\delta^{(N)}) - V_0(\delta) &= \int_{0}^{t} \delta^{(N)}_s \vd S_s,
	\end{align}
	where from Theorem~\ref{thm_prepricing},
	the option premium is
	$V_0(\delta)= \Esp^{\Qrob} \bigbraces{h(S_T)} $
	and
	$V_t(\delta)= \Esp^{\Qrob} \bigbraces{h(S_T)| \bF_t} = Q^{S,\Qrob}_{S_t} \bigbraces{h(Y_{T-t})} $.
	
	This discrete replication strategy, outlined in Algorithm~\ref{algo_delta_hedging_sticky}, results in replication errors known as the \textit{tracking error}, defined by 
	\begin{equation}
		\varepsilon^{(N)}_t := V_{t}(\delta^{(N)})-V_t(\delta).
	\end{equation}
	In~\cite{Bertsimas2001}, the metric used to assess the replication quality is the mean squared tracking error, defined by   
	\begin{equation}
		\mathsf R_T^{(N)} := \sqrt{\Esp^{\Qrob}\bigsqbraces{(\varepsilon^{(N)}_T)^{2}}}.
	\end{equation}

	Instead of this, we will quantify the replication properties of each model with the pair $\bigbraces{\Esp^{\Qrob} \sqbraces{\varepsilon^{(N)}_t},\bigbraces{\Var^{\Qrob} \sqbraces{\varepsilon^{(N)}_t}}^{1/2}} $, where $\Esp^{\Qrob} \sqbraces{\varepsilon^{(N)}_t} $ is the expected replication PnL and $\bigbraces{\Var^{\Qrob} \sqbraces{\varepsilon^{(N)}_t}}^{1/2} $ is the replication PnL standard deviation. The reason is that $\mathsf R_t^{(N)}$ is more adapted to the case where $\Esp^{\Qrob}\bigsqbraces{\varepsilon^{(N)}_t} =0$. If an agent is unsure about the pricing/replication model he applies, he is likely to miscalculate the option premium and end up with non-zero PnL, making the metric $\mathsf R_t^{(N)}$ irrelevant. This is the object of Section~\ref{ssec_model_mismatch}. Another reason for this choice is that the pair 
	$\bigbraces{\Esp^{\Qrob} \sqbraces{\varepsilon^{(N)}_t}
		,\bigbraces{\Var^{\Qrob} \sqbraces{\varepsilon^{(N)}_t}}^{1/2}} $
	contains the information in $\mathsf R_t^{(N)}$. Indeed, the three quantities are linked with the relation:
	\begin{equation}
		\Var^{\Qrob} \sqbraces{\varepsilon^{(N)}_t} 
		= \bigbraces{\mathsf R_t^{(N)}}^{2}
		- \bigbraces{\Esp^{\Qrob} \sqbraces{\varepsilon^{(N)}_t}}^{2}.
	\end{equation}

	\begin{table}
		\begin{center}
			\begin{tabular}{@{}l|rrr@{}}
				\multicolumn{4}{c}{Stickiness $\rho=0 $} \\
				N  & \multicolumn{1}{c}{premium} &\multicolumn{1}{c}{ $\widehat \mu_{\MC} $} &\multicolumn{1}{c}{ $\widehat{ \sigma}^{(N)}_{\MC} $ }\\ 	
				\hline 
				2000 & 3.03923 & -0.03438 & 0.06251 \\ 
				1000 & 3.03923 & -0.03481 & 0.08605 \\ 
				250 & 3.03923 & -0.03687 & 0.16781 \\ 
				100 & 3.03923 & -0.04529 & 0.26832 \\ 
				10 & 3.03923 & -0.00601 & 0.69078
			\end{tabular}
			\quad~\quad
			\begin{tabular}{@{}l|rrr@{}}
				\multicolumn{4}{c}{Stickiness $\rho=1 $} \\
				N  & \multicolumn{1}{c}{premium} &\multicolumn{1}{c}{ $\widehat \mu_{\MC} $} &\multicolumn{1}{c}{ $\widehat{ \sigma}^{(N)}_{\MC} $ }\\ 	
				\hline 
				2000 & 2.43607 & -0.02426 & 0.08783 \\ 
				1000 & 2.43607 & -0.01573 & 0.17748 \\ 
				250 & 2.43607 & -0.01959 & 0.29477 \\ 
				100 & 2.43607 & -0.02013  & 0.40177 \\ 
				10 & 2.43607 & -0.00182  & 0.80982
			\end{tabular}
			\\[10pt]
			\begin{tabular}{@{}l|rrr@{}}
				\multicolumn{4}{c}{Stickiness $\rho=2 $} \\
				N  & \multicolumn{1}{c}{premium} &\multicolumn{1}{c}{ $\widehat \mu_{\MC} $} &\multicolumn{1}{c}{ $\widehat{ \sigma}^{(N)}_{\MC} $ }\\ 	
				\hline 
				2000 & 1.9976 & -0.01487  & 0.04602 \\ 
				1000 & 1.9976 & -0.01502  & 0.18474 \\ 
				250 & 1.9976 & -0.03455  & 0.38181 \\ 
				100 & 1.9976 & -0.03731  & 0.51832 \\ 
				10 & 1.9976 & 0.01637  & 0.96878
			\end{tabular}
		\end{center}
		\caption{
			\centering	
			\small
			Premiums and Monte Carlo estimation of the expected replication error (1000 simulated trajectories) 
			for a Call $(K,T)=(10,10) $
			for various stickiness and hedging windows.
			$N_{\MC}=1000 $.
		}
		\label{table_granularity}
	\end{table}

	\begin{figure}
		\begin{center}
			\includegraphics[alt={Discrete time hedging non-sticky},width = 0.8\textwidth]{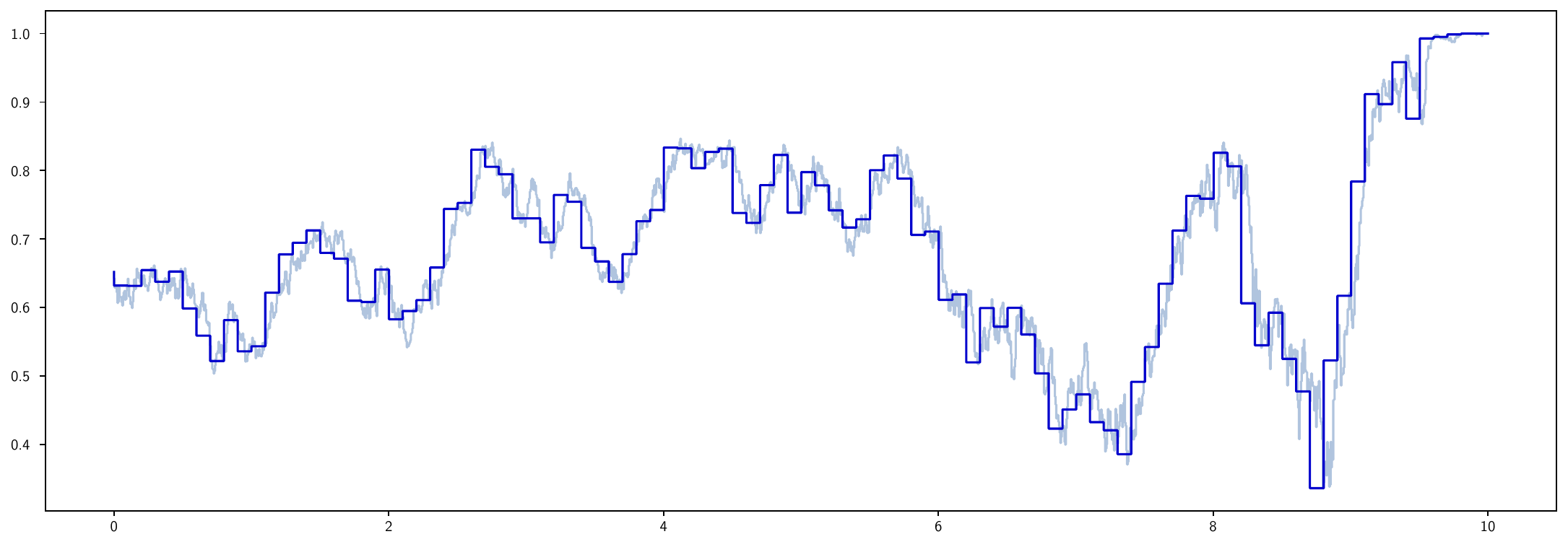}\\
			\includegraphics[alt={Discrete time hedging sticky},width = 0.8\textwidth]{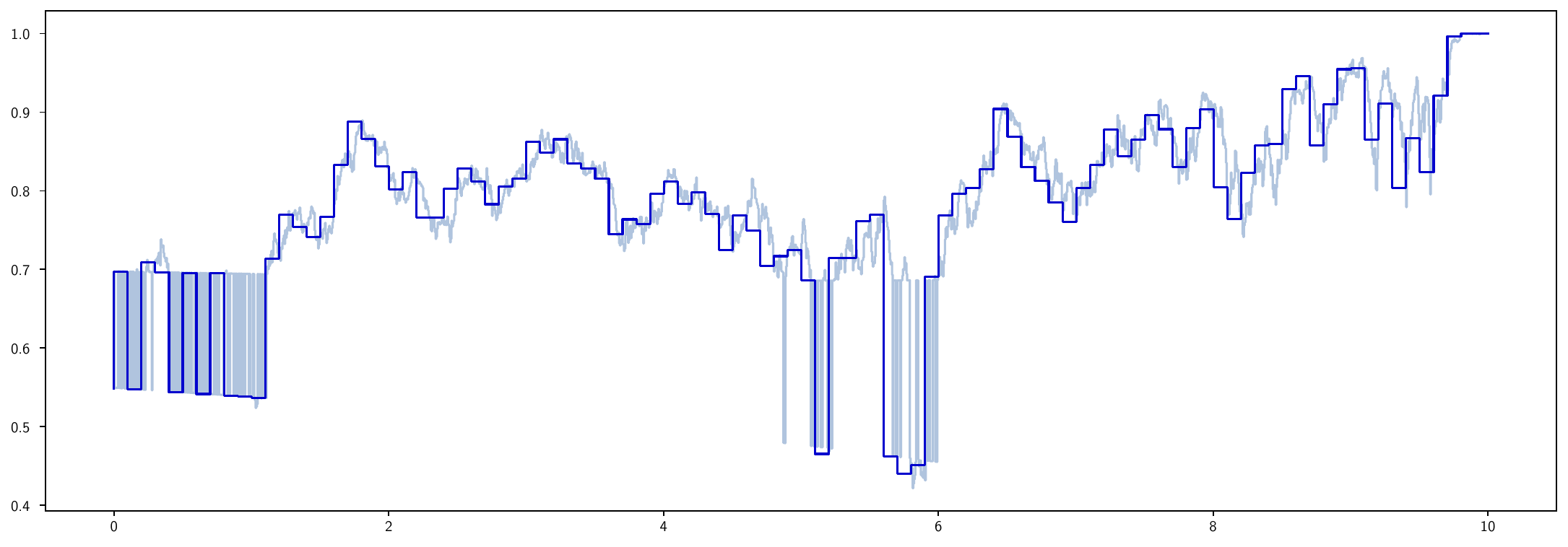}
			\caption{\small
				\centering 
				(top): Parts of risky asset in the payoff replication strategy for one sample path of $S$ with $\rho=0 $, for (dark blue) $N=100$ and (light blue) $N=2000$.
				(bottom): Same, but with $\rho=1 $.
			}
			\label{fig:time_granularityB}
		\end{center}
	\end{figure}

	We assess the discrete-time hedging error of the sticky Black-Scholes model as follows. 
	We sample trajectories of the sticky geometric Brownian motion.
	Then, we apply Algorithm~\ref{algo_delta_hedging_sticky} to hedge a call option 
	for various hedging windows.
	Last, we compute the empiric mean and standard deviation
	of the PnL with 
	\begin{align}
		\widehat \mu^{(N)}_{\MC} &= \frac{1}{N_{\MC}} \sum_{i=1}^{N_{\MC}} \varepsilon^{(N)}_{T}(\omega_i),
		&
		\widehat{ \sigma}^{(N)}_{\MC} &=  \frac{1}{\sqrt{N_{\MC}}} \sqrt{  \sum_{i=1}^{N_{\MC}} \bigbraces{\varepsilon^{(N)}_{T}(\omega_i) - \widehat \mu_{\MC}}^{2} }.
	\end{align}

	We have no closed form expressions for the replication strategy and the \APP{} prices of claims in the sticky Black-Scholes model. 
	Therefore, we use a finite difference approximation to numerically solve equation~\eqref{eq_def_BlackScholes_equ}.
	Simulation results are summarized in Table~\ref{table_granularity}.

	From numerical simulations, we observe the following:
	
	\begin{enumerate}
		\item The replication errors decreases the higher the hedging frequency $N$ is.
		Convergence seems to occur at the same rate as for the standard Black-Scholes model.
		In~\cite{Bertsimas2001}, it is proven that the rate of convergence for the Black-Scholes model is $\mathcal O (N^{-1/2}) $.
		\item Higher stickiness leads to higher tracking errors for coarse hedging time-windows. 
		We attribute this behavior to the inherent discontinuity of the delta function at the sticky threshold. 
		As depicted in Figure~\ref{figure_price_curves2}, higher stickiness at $\zeta$ corresponds to more substantial jumps in the delta, resulting in sudden shifts in portfolio composition whenever the price process $S$ crosses the sticky threshold. 
		Consequently, the discrete-time hedging portfolio manifests substantial deviations in portfolio composition from the ideal configuration whenever the price crosses the sticky threshold. 
		This is illustrated in Figure~\ref{fig:time_granularityB}, where spikes and oscillations in the portfolio composition for fine-grained hedging do not appear in coarser regimes. 
	\end{enumerate}

	\subsection{Model mismatch: constant volatility}\label{ssec_model_mismatch}
	
	Now, let's assume that an agent intends to hedge a contingent claim on an underlying asset with a sticky geometric Brownian motion price dynamic  with parameters $(\sigma,\rho,\zeta)$. 
	The agent considers the following models:

	\paragraph{Model 1:} The standard Black-Scholes model with volatility $\sigma $.
	
	\paragraph{Model 2:} The standard Black-Scholes model with volatility $ \Sigma_n (X) $, where for some $T'>0 $, $ \Sigma_n (X) $ is the estimator of the realized volatility on $[-T',0)$ of the current sample path. The estimator is defined by 
	\begin{equation}
		\Sigma_n (X) 
		= 
		\frac{1}{\sqrt{T'}}
		\sqrt{
			\sum_{i=1}^{[nT']} 
			\bigbraces{\log S_{-\frac{i}{n}} - \log S_{-\frac{i+1}{n}}}^{2}}.
	\end{equation}
	
	\paragraph{Model 3:} The sticky Black-Scholes model with parameters $(\sigma,\rho,\zeta) $.

	\begin{table}[t]
		\begin{center}
			\begin{tabular}{@{}c|llrrr@{}}
				\multicolumn{5}{c}{Model 1} \\
				$\rho$  & prem. & \APP & $\MP$  & \multicolumn{1}{c}{ $\widehat \mu_{\MC} $} &\multicolumn{1}{c}{ $\widehat{ \sigma}^{(N)}_{\MC} $ } 
				\\ 
				\hline 
				0 & 3.073 & 3.039 & 0.034 & -0.0018 & 0.057 \\ 
				1 & 3.073 & 2.436 & 0.637 & 0.6224 & 0.515 \\ 
				2 & 3.073 & 1.997 & 1.076 & 1.0551 & 0.727 
			\end{tabular}
			~
			\begin{tabular}{@{}c|lrrrr@{}}
				\multicolumn{5}{c}{Model 2} \\
				$\rho$  & prem. & \APP & $\MP$  &\multicolumn{1}{c}{ $\widehat \mu_{\MC} $} &\multicolumn{1}{c}{ $\widehat{ \sigma}^{(N)}_{\MC} $ } 
				  \\ 
				\hline 
				0 & 3.076 & 3.039  & 0.037  &  0.0016 & 0.077\\ 
				1 & 2.670 & 2.436 & 0.234 & 0.2172 & 0.508  \\ 
				2 & 2.359 & 1.997 & 0.362  & 0.3238 & 0.735
			\end{tabular}
			\\[10pt]
			\begin{tabular}{@{}c|lrrrr@{}}
				\multicolumn{5}{c}{Model 3} \\
				$\rho$  & prem. & \APP & $\MP$   &\multicolumn{1}{c}{ $\widehat \mu_{\MC} $} &\multicolumn{1}{c}{ $\widehat{ \sigma}^{(N)}_{\MC} $ }  \\ 
				\hline 
				0 & 3.039  & 3.039 & 0& -0.0343 & 0.062 \\ 
				1 & 2.436  & 2.436 & 0& -0.0242 & 0.087 \\ 
				2 & 1.997  & 1.997 & 0& -0.0014 & 0.046 \\
			\end{tabular}			
		\end{center}
		\caption{
			\small
			\centering	
			Theoretical premiums in each model, \APP{} prices, mean PnLs and tracking errors for a Call $(K,T)=(10,10) $
			on an assets with price-processes $\sigma=0.25 $ sticky at $\zeta=10 $ for various degrees of stickiness $\rho $ 
			in each model.
			The hedging frequency is $N= 2000$ and the number of simulated sample paths of each process is
			$N_{\MC}=1000 $.
		}
		\label{table_model_mismatch}
	\end{table}

	According to Theorem~\ref{thm_prepricing}, the replication cost for the claim with payoff $h(S_T)$ is $\Esp^{\Qrob}\sqbraces{h(S_T)} $. This matches the \APP{} price or premium in Model~3, i.e., 
	\begin{equation}
		\text{\APP{}} \text{ price} = \bigbraces{\text{Premium in Model 3}} = \Esp^{\Qrob}\sqbraces{h(S_T)}.
	\end{equation}
	Therefore, the premium is miscalculated in Models~1 and 2. To address the premium miscalculation in each model, we introduce the mispricing error, defined as 
	\begin{equation}
		\MP  :=  \bigbraces{\text{Premium in Model }i} - \text{\APP{}} \text{ price}.
	\end{equation}
	This error compensates for the error made by the agent when using an incorrect model to compute premiums for contingent claims.
	It is to be compared with the empiric mean PnL of the replication portfolios $\widehat \mu_{\MC}  $.
	The discrepency between values between Model~1 and Model~3 when $\rho=0 $ is due to error incurred from the finite difference approximations of quantities in Model~3.
	
	We observe the following: 
	
	\begin{enumerate}
		\item We observe in Models~1 that the mispricing error $\MP $ is close to the empiric mean PnL of the replication portfolios $\widehat \mu_{\MC}  $. 
		This is due to two facts. 
		First, the agent ultimately holds $h(S_T) $ at maturity. 
		Second, given that the price dynamics are martingales on $(\Omega,\process{\bF_t},\Qrob) $, from \eqref{eq_thm_replication} the self-financing portfolio is also a martingale.
		Hence, the expected PnL of any self-financing portfolio is $0$
		and the expected PnL of the option seller is
		\begin{equation}
			\text{Premium} -\Esp^{\Qrob} \sqbraces{h(S_T)},
		\end{equation}
		which equals $0$ only when pricing with the correct model, Model~3. 
		\item Among Models~1, 2 and 3, it is Model~3 that exhibits the lowest hedging error, indicating superior replication properties (see Table~\ref{table_model_mismatch}).
	\end{enumerate}

	\subsection{Model mismatch: smooth SDE}\label{ssec_model_mismatchB}

	We now suppose the agent misrepresents price stickiness as local volatility in a smooth SDE model.
	
	It is possible to approximate in law diffusions with sticky features by smooth diffusions.
	Let $X$ be the diffusion on $(0,\infty) $, on natural scale,
	of speed measure
	\begin{equation}
		\begin{aligned}
			m(\rd x)
			&= \frac{2}{\sigma^{2}x^{2}}  \vd x +  \rho  \delta_{\zeta}(\rd x),
			& x&>0.
		\end{aligned}
	\end{equation}
	It corresponds to the law of $S$ under $\Qrob $, see~\eqref{eq_prop_risk_neutral_diffusion}.
	We can define an approximating sequence $(X^{n})_n $ of $S$ in law as follows.
	For all $ n$: let $X^{n} $ solve $\vd X^{n}_t = \sigma_n(X^{n}_t) \vd B^{n}_t  $ with $B^{n} $ a standard Brownian motion and $\sigma_n $ a smooth function, such that  
	\begin{equation}
		\frac{2}{\sigma_{n}^{2}(x)}  \vd x
		\xrightarrow[n\rightarrow \infty]{\text{weakly}} m(\rd x).
	\end{equation}
	From \cite{brooks1982weak}, $X^{n}\lra X $ in law.
	
	We make two key observations.
	First, the prices in the smooth models converge to the prices in the sticky model.
	Indeed, let $v $ be the price function and $(Q_{x}^{X})_x$ the Markovian family of the sticky model $X$, and, for all $n\in \IN$, let $v^{(n)} $ be the price function 
	and $ (Q_{x}^{X^{n}})_x$ the Markovian family of the smooth model $X^{n} $.
	The convergence in law implies that 
	\begin{equation}
		\begin{aligned}
			\lim_{n\rightarrow \infty} v^{(n)}(t,x) 
			= \lim_{n\rightarrow \infty}  Q^{X^{n}}_x \bigbraces{h(Y_{T-t})}
			= Q^{X}_x \bigbraces{h(Y_{T-t})}
			= v(t,x),
		\end{aligned}
	\end{equation} 
	where $Y$ is the coordinate process. 
	Second, while the deltas in the smooth models are continuous, the delta of the sticky model is discontinuous (see Remark~\ref{rmk_discontinuity}).
	Indeed, due to smoothness of coefficient $\sigma_n $, the function $x \rightarrow v^{(n)}(t,x)$ is also smooth, for all $t\in [0,T] $ and $n\in \IN $.

	\begin{table}[h!]
		\begin{center}
			\begin{tabular}{@{}l|rrr@{}}
				\multicolumn{4}{c}{Stickiness $\rho=0 $} \\
				N  & \multicolumn{1}{c}{premium} &\multicolumn{1}{c}{ $\widehat \mu_{\MC} $} &\multicolumn{1}{c}{ $\widehat{ \sigma}^{(N)}_{\MC} $ }\\ 	
				\hline 
				2000 & 3.03923 & -0.03438 & 0.06251 \\ 
				1000 & 3.03923 & -0.03481 & 0.08605 \\ 
				250 & 3.03923 & -0.03687 & 0.16781 \\ 
				100 & 3.03923 & -0.04529 & 0.26832 \\ 
				10 & 3.03923 & -0.00601 & 0.69078
			\end{tabular}
			\quad~\quad
			\begin{tabular}{@{}l|rrr@{}}
				\multicolumn{4}{c}{Stickiness $\rho=1 $} \\
				N  & \multicolumn{1}{c}{premium} &\multicolumn{1}{c}{ $\widehat \mu_{\MC} $} &\multicolumn{1}{c}{ $\widehat{ \sigma}^{(N)}_{\MC} $ }\\ 	
				\hline 
				2000 & 2.43607 & -0.02311 & 0.17189 \\ 
				1000 & 2.43607 & -0.01909 & 0.19802 \\ 
				250 & 2.43607 & -0.02165 & 0.29681 \\ 
				100 & 2.43607 & -0.02050 & 0.38760 \\ 
				10 & 2.43607 & -0.01681 & 0.78004
			\end{tabular}
			\\[10pt]
			\begin{tabular}{@{}l|rrr@{}}
				\multicolumn{4}{c}{Stickiness $\rho=2 $} \\
				N  & \multicolumn{1}{c}{premium} &\multicolumn{1}{c}{ $\widehat \mu_{\MC} $} &\multicolumn{1}{c}{ $\widehat{ \sigma}^{(N)}_{\MC} $ }\\ 	
				\hline 
				2000 & 1.9976 & -0.00126 & 0.24402 \\ 
				1000 & 1.9976 & -0.00939 & 0.28132 \\ 
				250 & 1.9976 & -0.02760 & 0.40503 \\ 
				100 & 1.9976 & -0.02075 & 0.52987 \\ 
				10 & 1.9976 & 0.00303 & 0.94837
			\end{tabular}
		\end{center}
		\caption{
			\centering	
			\small
			Same as Table~\ref{table_granularity} but for the locally linearly interpolated price curve around $\zeta $.
		}
		\label{table_granularityM}
	\end{table}

	In Section~\ref{ssec_discrete_time_hedging}, we approximated the delta function using an interpolation of the finite-difference scheme that preserved the discontinuity.
	In this section, we use an approximation that does not preserve the discontinuity.
	Thus, the delta for the two approximations differ only locally around the threshold of stickiness.
	We observe that the second method produces high tracking error (see  Table~\ref{table_granularityM}, to compare with Table~\ref{table_granularity}).
	This indicates that hedging a sticky diffusion using a smooth model results in some significative irreductible tracking error. 
	The reason is that a smooth function of the spot price cannot replicate the discontinuity of the true sticky hedging strategy.

	\section{Conclusion}
	\label{sec_conclusion}
	
	In conclusion, we have demonstrated that the sticky Black-Scholes model satisfies the No Free Lunch with Vanishing Risk \NFLVR{} and is consistent with arbitrage-free pricing theory only when the interest rate $r=0 $. Under this condition, we proved the existence of a replication strategy and that in this model there is no uniqueness in both the replication strategy and the equivalent local martingale measure \ELMM{}. Moreover, we established that any locally bounded payoff that can be replicated in the standard Black-Scholes model can also be replicated in the sticky model.
	
	Our numerical experiments indicate that the mean squared tracking error of discrete-time hedging converges at a rate of 
	$\mathcal O(N^{1/2}) $, to $0$ as $N$ diverges to $\infty $, matching the rate in the standard Black-Scholes model. We also observed that ignoring or misrepresenting stickiness results in significant irreducible tracking errors, emphasizing the importance of using hedging strategies that accurately account for sticky price features.
	
	This work leaves several topics for future investigation. First, exploring whether the model with $r \not = 0 $ remains arbitrage-free under proportional or fixed transaction costs. Second, identifying the class of payoffs for which the pricing problem stated in Theorem~\ref{thm_pricing} admits a unique classical solution. Lastly, determining the rate of convergence of the mean squared tracking error that arises from discrete-time hedging.
	
	\appendix

	\section{Proof--assisting results}
	\label{app_proofs}

	\subsection{Proof of Lemma~\ref{lem_timechange}}
	\label{app_lem_timechange}

	For reader's convenience, we recall the statement of Lemma~\ref{lem_timechange}.
	
	\begin{lemmaOhne}
		Let $X$ be a regular diffusion on $J $, an open interval of $\IR $, on natural scale of speed measure $m$ so that 
		$\rho := m(\{\zeta\})\in (0,\infty) $.
		We suppose $X$ to be defined on the probability space $\mc P_x = (\Omega,\process{\bF_t},\Prob_x) $.
		There exists a regular diffusion $Z $ on $ J$, on natural scale, with speed measure $m_Z = m - \rho \delta_{\zeta} $, defined on an extension of $\mc P_x $,  such that 
		\begin{equation}
			\forall t\ge 0:\qquad
			X_{t} = Z_{\gamma_{\rho}(t)},
			\qquad A_{\rho}(t) = t + \rho \loct{Z}{\zeta}{t},
		\end{equation}
		where $\gamma_{\rho} $ is the right-inverse of $A_{\rho} $.
	\end{lemmaOhne}
	
	\begin{proof}
		From~\cite[Theorem~V.47.1]{RogWilV2}, there exists a standard Brownian motion $W$, defined on an extension of $\mc P_x $, 
		such that 
		\begin{equation}
			\forall t\ge 0:\qquad
			X_{t} = W_{\gamma_X(t)}
			\qquad A_{X}(t) = \int_{J} \loct{W}{y}{t} m(\rd y), 
		\end{equation}
		where $\gamma_X $ is the right-inverse of $A_{X} $. 
		We consider the process $Z $ defined by 
		\begin{equation}
			\forall t\ge 0:\qquad
			Z_{t} = W_{\gamma_Z(t)}
			\qquad A_{Z}(t) = \int_{J} \loct{W}{y}{t} m_{Z}(\rd y),
		\end{equation}
		where $\gamma_Z $ is the right-inverse of $A_{Z} $.
		From~\cite[Section~V.47, Remark~(ii)]{RogWilV2}, $Z$ is a diffusion on $J $ on natural scale with speed measure $m_Z $.
		Thus, $X = \process{Z_{ A_Z \circ \gamma_X   (t)}} $.
		From \cite[Proposition~IV.1.12]{RevYor} and the proof of~\cite[Theorem~V.47.1]{RogWilV2}, for the case of a regular diffusion with unattainable boundaries, we have that $ \gamma_Z(\cdot) = \qv{Z}$ and that $\qv{Z} $ is almost surely strictly increasing. 
		Thus, the time--changes $\gamma_Z,A_Z,\gamma_X,A_X $ are all almost surely strictly increasing and hence invertible,  
		\ie{} $A_Z = \gamma_{Z}^{-1}$, $A_X  = \gamma_{X}^{-1} $.
		This with~\cite[Exercise~VI.1.27]{RevYor} yield that  
		\begin{equation}
			\gamma_{X}^{-1}(t) = \int_{J} \loct{W}{y}{t} m(\rd y)
			= \int_{J} \loct{W}{y}{t} m_Z(\rd y) + \loct{W}{\zeta}{t} m(\{\zeta\})
			= \gamma^{-1}_{Z}(t) +  m(\{\zeta\}) \loct{Z}{\zeta}{\gamma^{-1}_Z(t)}
		\end{equation}
		and that
		\begin{equation}
			\bigbraces{\gamma^{-1}_{Z} \circ \gamma_X}^{-1} (t)
			= \gamma^{-1}_X \circ \gamma_{Z} = t + m(\{\zeta\}) \loct{Z}{\zeta}{t} = A_{\rho}(t).
		\end{equation}
		This completes the proof. 
	\end{proof}
	
	\subsection{Proof of Lemma~\ref{lem_kernel_monotony}}
	\label{app_lem_kernel_monotony}
	
	For reader's convenience, we recall the statement of Lemma~\ref{lem_kernel_monotony}.
	
	\begin{lemmaOhne}
		Let $[(t,x,y)\ra p_0(t,x,y)] $ and $m_0(\rd y) = \dot m_0( y) \vd y $ be the probability transition kernel and speed measure of the  geometric Brownian motion.
		For all $x>0 $ and $T>0 $, there exists a compact $K_{x,T}$ of $(0,\infty) $ such that on $(t,y)\in [0,T] \times K^{c}_{x,T} $, the function $[t\ra p_0(t,x,y) \dot m_0(y)] $ is increasing.
	\end{lemmaOhne}
	
	\begin{proof}
		From the expressions for the probability transition kernel and speed measure of the geometric Brownian motion
		(see~\cite[Appendix~1.20]{BorSal}):
		\begin{equation}
			\dot m(y) p_{0}(t,x,y)
			= 
			\frac{1}{y \sigma \sqrt{2\pi t}}
			\exp \biggbraces{ - \frac{\bigbraces{\log y - \log x - \bigbraces{\mu - \frac{1}{2} \sigma^{2}}t}^{2}  }{2 \sigma^{2} t}}.
		\end{equation}
		Taking the derivative in time yields 
		\begin{equation}
			\begin{aligned}
				\partial_t \bigbraces{\dot m(y) p_{0}(t,x,y)}
				= &  \biggbraces{ \frac{\bigbraces{\log y - \log x - \bigbraces{\mu - \frac{1}{2} \sigma^{2}}t}^{2}}{2 \sigma^{2} t^{2}}
				\\	&+ 2 \frac{\bigbraces{\log y - \log x - \bigbraces{\mu - \frac{1}{2} \sigma^{2}}t}}{2 \sigma^{2} t}\bigbraces{\mu - \frac{1}{2} \sigma^{2}}t - \frac{1}{2t}} \dot m(y) p_{0}(t,x,y).
			\end{aligned}
		\end{equation}
		We observe that for all $t,x >0$, as $y \lra 0$ and $y\lra \infty $, the quantity $\partial_t p_{0}(t,x,y) $
		becomes positive.
		This completes the proof.
	\end{proof}
	
	\begin{small}

		\bibliographystyle{abbrv}
	\end{small}
	
\end{document}